\documentclass[a4paper,11pt]{amsart}
\usepackage{amssymb,amsmath,amsfonts,amsthm}
\usepackage[english]{babel}
\usepackage[yyyymmdd]{datetime}
\usepackage{enumerate,expdlist}
\usepackage{tikz}
\usepackage{caption}
\usepackage{subcaption}

\usepackage{relsize} 
\newcommand{\euler}{\mathrm{e}}
\newcommand{\RR}{\mathbb{R}}
\newcommand{\CC}{\mathbb{C}}
    
\newcommand{\ZZ}{\mathbb{Z}}

\newcommand{\TT}{\mathbb{T}}

\newcommand{\vol}{\operatorname{Vol}}

\newcommand{\drm}{\mathrm{d}}
\newcommand{\tr}{\operatorname{Tr}}
\newcommand{\id}{\operatorname{Id}}

\newcommand{\largediamond}{\mathlarger{\mathlarger{\mathlarger{\diamond}}}}

\newcommand{\hex}{\mathrm{(Hex)}}
\newcommand{\tri}{\mathrm{(Tri)}}

\newcommand{\V}{\mathcal{V}}
\newcommand{\E}{\mathcal{E}}

\newtheorem{theorem}{Theorem}[section]
\newtheorem{lemma}[theorem]{Lemma}
\newtheorem{proposition}[theorem]{Proposition}
\newtheorem{corollary}[theorem]{Corollary}

\theoremstyle{definition}
\newtheorem{definition}[theorem]{Definition}
\newtheorem{remark}[theorem]{Remark}
\usepackage{microtype}
\usepackage{ellipsis}
\usepackage{typearea}
\usepackage[textsize=tiny,shadow]{todonotes}
\begin{document}
\title[Eigenfunctions and the IDS of Archimedean Tilings]{Eigenfunctions and the Integrated Density of States on Archimedean Tilings}
\author{Norbert Peyerimhoff and Matthias T\"aufer}
\address[NP]{Department of Mathematical Sciences, Durham University, UK}
\address[MT]{School of Mathematical Sciences, Queen Mary University of London, UK}
\keywords{eigenfunctions -- Archimedean tilings -- Floquet Theory  -- Integrated density of states}
\date{}

\begin{abstract}
 We study existence and absence of $\ell^2$-eigenfunctions of the
 combinatorial Laplacian on the $11$ Archimedean tilings of the
 Euclidean plane by regular convex polygons. We show that exactly two
 of these tilings (namely the $(3.6)^2$ ``Kagome'' tiling and the
 $(3.12^2)$ tiling) have $\ell^2$-eigenfunctions.  These
 eigenfunctions are infinitely degenerate and are constituted of
 explicitly described eigenfunctions which are supported on a finite
 number of vertices of the underlying graph (namely the hexagons and
 $12$-gons in the tilings, respectively).
 Furthermore, we provide an explicit expression for the Integrated
 Density of States (IDS) of the Laplacian on Archimedean tilings in
 terms of eigenvalues of Floquet matrices and deduce integral formulas
 for the IDS of the Laplacian on the $(4^4)$, $(3^6)$, $(6^3)$,
 $(3.6)^2$, and $(3.12^2)$ tilings.
 Our method of proof can be applied to
 other $\ZZ^d$-periodic graphs as well.
\end{abstract}

\maketitle

\section{Introduction and statement of results}

The goal of this paper is to provide concrete formulas for the
Integrated Density of States (IDS) on Archimedean tilings,
viewed as combinatorial graphs, and to study existence or absence of
$\ell^2$-eigenfunctions for the associated Laplacians. 

A \emph{plane tiling by regular convex polygons} is a countable family of
regular convex polygons covering the plane without gaps or overlaps.  It is
called \emph{edge-to-edge} if the corners and sides of the polygons
coincide with the vertices and edges of the tiling
(see~\cite{GruenbaumS-89}).  The \emph{type} of a vertex of an
edge-to-edge plane tiling by regular polygons describes the order of
the polygons arranged cyclically around the vertex, for example the
vertices in the honeycomb tiling are all of the type $(6.6.6) =: (6^3)$.

\begin{definition}
 An \emph{Archimedean tiling} is an edge-to-edge tiling of the plane
 by regular convex polygons such that all vertices are of the same type.
\end{definition}

Archimedean tilings were systematically investigated in 1619 by
Johannes Kepler in his book Harmonices Mundi \cite{Ke1619} (see
\cite{Fi79} for an English translation). Kepler found all $11$
Archimedean tilings, namely with vertices of type $(4^4)$, $(3^6)$,
$(6^3)$, $(3.6)^2$, $(3.12^2)$, $(4.8^2)$, $(3^3.4^2)$, $(3^2.4.3.4)$,
$(3.4.6.4)$, $(3^4.6)$, and $(4.6.12)$, cf.~\cite[p.~59,
  63]{GruenbaumS-89} and Figure~\ref{fig:11_tilings} for an illustration.

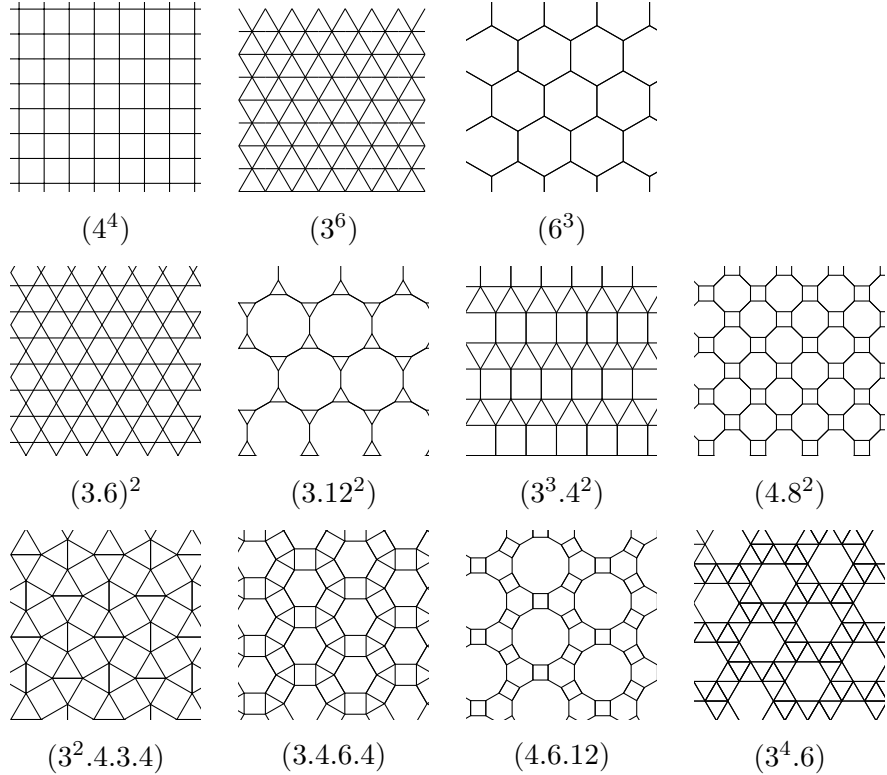
\begin{figure}[ht]
 \begin{tikzpicture}
  
\begin{scope}
\draw (1.25,-.5) node {$(4^4)$};

 \clip (0,0) rectangle (2.5,2.5); 
\begin{scope}[scale =.33, xshift = .33cm, yshift = .33cm]
   \draw (-1, -1) grid (8,8);
\end{scope}
\end{scope}
  
\begin{scope}[xshift = 3cm]
\draw (1.25,-.5) node {$(3^6)$};
 \clip (0,0) rectangle (2.5,2.5);
\begin{scope}[scale =.35]
  \foreach \i in {0,...,6}
    {
    \foreach \j in {0,...,3}
        {
        \pgfmathsetmacro{\jj}{2*\j*0.86602540378}
        \draw (\i,\jj) -- (\i+1,\jj);
        \draw (\i,\jj) -- (\i+1, \jj+2*0.86602540378);
        \draw (\i+1,\jj) -- (\i, \jj + 2*0.86602540378);
        \draw (\i, \jj + 0.86602540378) -- (\i + 1,\jj + 0.86602540378);
        }
    }
\end{scope}
\end{scope}

\begin{scope}[xshift = 6cm]
\draw (1.25,-.5) node {$(6^3)$};
 \clip (0,0) rectangle (2.5,2.5);
\begin{scope}[scale = .4, xshift = -.05cm]
  \foreach \i in {0,...,6}
    {
    \foreach \j in {0,...,3}
        {
        \pgfmathsetmacro{\ii}{2*\i*0.86602540378}
        \pgfmathsetmacro{\jj}{3*\j}
        \begin{scope}[xshift = \ii cm, yshift = \jj cm]
                \draw (30:1) -- (90:1) -- (150:1) -- (210:1) -- (270:1) -- (330:1) -- (30:1);
            
            \begin{scope}[xshift = 0.86602540378 cm , yshift = 1.5 cm]
                \draw (30:1) -- (90:1) -- (150:1) -- (210:1) -- (270:1) -- (330:1) -- (30:1);
            \end{scope}
        \end{scope}
        }
    }
\end{scope}
\end{scope}

\begin{scope}[yshift = -3.5cm]
\draw (1.25,-.5) node {$(3.6)^2$};
 \clip (0,0) rectangle (2.5,2.5);
\begin{scope}[scale = .2]
  \foreach \i in {0,...,6}
    {
    \foreach \j in {0,...,3}
        {
        \pgfmathsetmacro{\ii}{2*\i}
        \pgfmathsetmacro{\jj}{4*\j*0.86602540378}
        \begin{scope}[xshift = \ii cm, yshift = \jj cm]
                \draw (0:1) -- (60:1) -- (120:1) -- (180:1) -- (240:1) -- (300:1) -- (360:1);
                
            \begin{scope}[xshift = 1cm, yshift = 2*0.86602540378cm]
                \draw (0:1) -- (60:1) -- (120:1) -- (180:1) -- (240:1) -- (300:1) -- (360:1);
            \end{scope}
        \end{scope}
        }
    }
\end{scope}
\end{scope}    
    
\begin{scope}[xshift = 3cm, yshift = -3.5cm]
\draw (1.25,-.5) node {$(3.12^2)$};
 \clip (0,0) rectangle (2.5,2.5);
\begin{scope}[scale = .22] 
   \foreach \i in {0,...,6}
    {
    \foreach \j in {0,...,3}
        {
        \pgfmathsetmacro{\ii}{\i * (2*0.86602540378 + 2)}
        \pgfmathsetmacro{\jj}{\j * (4 * 0.86602540378 + 3)}
        \begin{scope}[xshift = \ii cm, yshift = \jj cm]
        \draw (0,0) -- (1,0) -- (.5,0.86602540378) -- (0,0);
        \draw (.5,0.86602540378) -- (.5,1.86602540378);
        \draw (.5,1.86602540378) -- (1,1.86602540378 + 0.86602540378) -- (0,1.86602540378 + 0.86602540378) -- (.5,1.86602540378);
            \begin{scope}
            \draw (210:1) -- (0,0);
            \end{scope}
            \begin{scope}[xshift = 1cm]
            \draw (330:1) -- (0,0);
            \end{scope}
            \pgfmathsetmacro{\l}{2*0.86602540378 + 1}
            \begin{scope}[yshift = \l cm]
            \draw (150:1) -- (0,0);
            \end{scope}
            \begin{scope}[xshift=1cm, yshift = \l cm]
            \draw (30:1) -- (0,0);
            \end{scope}
        \end{scope}
        
        \pgfmathsetmacro{\iii}{\ii + (0.86602540378 + 1)}
        \pgfmathsetmacro{\jjj}{\jj + (2 * 0.86602540378 + 1.5)}
        
        \begin{scope}[xshift = \iii cm, yshift = \jjj cm]
        \draw (0,0) -- (1,0) -- (.5,0.86602540378) -- (0,0);
        \draw (.5,0.86602540378) -- (.5,1.86602540378);
        \draw (.5,1.86602540378) -- (1,1.86602540378 + 0.86602540378) -- (0,1.86602540378 + 0.86602540378) -- (.5,1.86602540378);
            \begin{scope}
            \draw (210:1) -- (0,0);
            \end{scope}
            \begin{scope}[xshift = 1cm]
            \draw (330:1) -- (0,0);
            \end{scope}
            \pgfmathsetmacro{\l}{2*0.86602540378 + 1}
            \begin{scope}[yshift = \l cm]
            \draw (150:1) -- (0,0);
            \end{scope}
            \begin{scope}[xshift=1cm, yshift = \l cm]
            \draw (30:1) -- (0,0);
            \end{scope}
        \end{scope}
        }
    }
\end{scope}
\end{scope}

\begin{scope}[xshift=6cm, yshift=-3.5cm]
\draw (1.25,-.5) node {$(3^3.4^2)$};
 \clip (0,0) rectangle (2.5,2.5);
\begin{scope}[scale = .4, xshift = -0.05cm]

  \foreach \i in {-1,...,6}
    {
    \foreach \j in {0,...,2}
        {
        \pgfmathsetmacro{\jj}{\j * 2 * (1+0.86602540378)}
        \begin{scope}[xshift = \i cm, yshift = \jj cm]
            \draw (0,0) -- (0,1) -- (1,1) -- (1,0) -- (0,0);
            \draw(0,1) -- (0.5,1+0.86602540378) -- (1,1);
             \pgfmathsetmacro{\l}{1+0.86602540378}
                \begin{scope}[xshift = .5cm, yshift = \l cm]
                    \draw (0,0) -- (0,1) -- (1,1) -- (1,0) -- (0,0);
                    \draw(0,1) -- (0.5,1+0.86602540378) -- (1,1);
                \end{scope}
        \end{scope}
        }
    }
\end{scope}
\end{scope}

\begin{scope}[xshift=9cm, yshift=-3.5cm]
\draw (1.25,-.5) node {$(4.8^2)$};
 \clip (0,0) rectangle (2.5,2.5);
\begin{scope}[scale = .2, xshift = .25cm]
 \foreach \i in {0,...,3}
    {
    \foreach \j in {0,...,3}
        {
        \pgfmathsetmacro{\ii}{\i * 2 * (1+0.70710678118)}
        \pgfmathsetmacro{\jj}{\j * 2 * (1+0.70710678118)}
        \begin{scope}[xshift = \ii cm, yshift = \jj cm]
            \draw (0,0) -- (0,1) -- (1,1) -- (1,0) -- (0,0);
            \draw (0,0) -- (-0.70710678118,-0.70710678118);
            \draw (1,0) -- (1.70710678118,-0.70710678118);
            \draw (0,1) -- (-0.70710678118,1.70710678118);
            \draw (1,1) -- (1.70710678118,1.70710678118);
            \pgfmathsetmacro{\l}{1+0.70710678118}
                \begin{scope}[xshift = \l cm, yshift = \l cm]
                    \draw (0,0) -- (0,1) -- (1,1) -- (1,0) -- (0,0);
                    \draw (0,0) -- (-0.70710678118,-0.70710678118);
                    \draw (1,0) -- (1.70710678118,-0.70710678118);
                    \draw (0,1) -- (-0.70710678118,1.70710678118);
                    \draw (1,1) -- (1.70710678118,1.70710678118);
                \end{scope}
        \end{scope}
        }
    }
\end{scope}
\end{scope}

\begin{scope}[yshift = -7cm]
\draw (1.25,-.5) node {$(3^2.4.3.4)$};
 \clip (0,0) rectangle (2.5,2.5);
\begin{scope}[scale = .4]
 \foreach \i in {0,...,3}
    {
    \foreach \j in {0,...,3}
        {
        \pgfmathsetmacro{\l}{1 + 2* 0.86602540378}
        \pgfmathsetmacro{\ii}{\i * \l}
        \pgfmathsetmacro{\jj}{\j * \l}
        \begin{scope}[xshift = \ii cm, yshift = \jj cm]
            \pgfmathsetmacro{\ll}{0.86602540378 + 0.5}
            \begin{scope}[xshift = 0cm, yshift = 0cm]
            \draw (0,0) -- (1,0) -- (.5,-0.86602540378) -- (0,0);
            \draw (0,0) -- (1,0) -- (.5,0.86602540378) -- (0,0);  
            \end{scope}
            \begin{scope}[xshift = .5cm, yshift = 0.86602540378 cm]
                \draw (0,0) -- (0,1) -- (-0.86602540378, .5) -- (0,0);
                \draw (0,0) -- (0,1) -- (0.86602540378,.5) -- (0,0);
            \end{scope}
            \begin{scope}[xshift = \ll cm, yshift = \ll cm]
            \draw (0,0) -- (1,0) -- (.5,-0.86602540378) -- (0,0);
            \draw (0,0) -- (1,0) -- (.5,0.86602540378) -- (0,0);  
            \end{scope}
            \begin{scope}[xshift = 1.86602540378 cm, yshift = -.5cm]
                \draw (0,0) -- (0,1) -- (-0.86602540378, .5) -- (0,0);
                \draw (0,0) -- (0,1) -- (0.86602540378,.5) -- (0,0);
            \end{scope}
        \end{scope}
        }
    }
\end{scope}
\end{scope}

\begin{scope}[xshift = 3cm, yshift = -7cm]
\draw (1.25,-.5) node {$(3.4.6.4)$};
 \clip (0,0) rectangle (2.5,2.5);
\begin{scope}[scale = .28, xshift = .25cm, yshift = .25cm]
 \foreach \i in {0,...,3}
    {
    \foreach \j in {0,...,5}
        {
        \pgfmathsetmacro{\ii}{\i * 2 * (1 + 0.86602540378 + 0.5)}
        \pgfmathsetmacro{\jj}{\j * 2 * (1 + 0.86602540378 - 0.5)}
        \begin{scope}[xshift = \ii cm, yshift = \jj cm]
            \draw (0,0) -- (0,1) -- (1,1) -- (1,0) -- (0,0);
            \draw (0,0) -- (-0.86602540378, .5) -- (0,1);
            \draw (1,0) -- (1.86602540378, .5) -- (1,1);
            \draw (0,0) -- (-.5,-0.86602540378);
            \draw (0,1) -- (-.5,1.86602540378);
            \draw (1,0) -- (1.5,-0.86602540378);
            \draw (1,1) -- (1.5,1.86602540378);
            \begin{scope}[xshift = -0.86602540378 cm, yshift = .5 cm]
                \draw (0,0) -- (-.5,-0.86602540378);
                \draw (0,0) -- (-.5,0.86602540378);
            \end{scope}
            \begin{scope}[xshift = 1.86602540378 cm, yshift = .5 cm]
                \draw (0,0) -- (.5,-0.86602540378);
                \draw (0,0) -- (.5,0.86602540378);
            \end{scope}
             \pgfmathsetmacro{\lx}{1 + 0.86602540378 + 0.5}
             \pgfmathsetmacro{\ly}{1 + 0.86602540378 - 0.5}
                \begin{scope}[xshift = \lx cm, yshift = \ly cm]
                    \draw (0,0) -- (0,1) -- (1,1) -- (1,0) -- (0,0);
                    \draw (0,0) -- (-0.86602540378, .5) -- (0,1);
                    \draw (1,0) -- (1.86602540378, .5) -- (1,1);
                    \draw (0,0) -- (-.5,-0.86602540378);
                    \draw (0,1) -- (-.5,1.86602540378);
                    \draw (1,0) -- (1.5,-0.86602540378);
                    \draw (1,1) -- (1.5,1.86602540378);
                    \begin{scope}[xshift = -0.86602540378 cm, yshift = .5 cm]
                        \draw (0,0) -- (-.5,-0.86602540378);
                        \draw (0,0) -- (-.5,0.86602540378);
                    \end{scope}
                    \begin{scope}[xshift = 1.86602540378 cm, yshift = .5 cm]
                        \draw (0,0) -- (.5,-0.86602540378);
                        \draw (0,0) -- (.5,0.86602540378);
                    \end{scope}
            \end{scope}
            
        \end{scope}
        }
    }
\end{scope}
\end{scope}

\begin{scope}[xshift = 6cm, yshift = -7cm]
\draw (1.25,-.5) node {$(4.6.12)$};
 \clip (0,0) rectangle (2.5,2.5);
\begin{scope}[scale = .2, xshift = .25cm, yshift = .25cm]
  \foreach \i in {0,...,2}
    {
    \foreach \j in {0,...,3}
        {
            \pgfmathsetmacro{\ii}{\i * 2*(1.5 + 3* 0.86602540378)} 
            \pgfmathsetmacro{\jj}{\j* 2*(1.5 + 0.86602540378)}

        \begin{scope}[xshift = \ii cm, yshift = \jj cm]
            \draw (0,0) -- (0,1) -- (1,1) -- (1,0) -- (0,0);
            \begin{scope}[xshift = -0.86602540378cm, yshift = .5cm]
                \draw (30:1) -- (90:1) -- (150:1) -- (210:1) -- (270:1) -- (330:1) -- (30:1);
            \end{scope}
            \draw (-0.86602540378, 1.5) -- (-0.86602540378-.5, 1.5+0.86602540378); 
            \draw (-2* 0.86602540378, 1) -- (-2*0.86602540378-.5,1.86602540378); 
            \draw (-0.86602540378, -.5) -- (-0.86602540378-.5, -.5-0.86602540378); 
            \draw (-2* 0.86602540378, 0) -- (-2*0.86602540378-.5,-.86602540378); 
            \begin{scope}[xshift = 1.86602540378cm, yshift = .5cm]
                \draw (30:1) -- (90:1) -- (150:1) -- (210:1) -- (270:1) -- (330:1) -- (30:1);
            \end{scope}
                \draw (1.86602540378, 1.5) -- (1.86602540378+.5, 1.5+0.86602540378); 
                \draw (1+2* 0.86602540378, 1) -- (1+2*0.86602540378+.5,1.86602540378); 
                \draw (1.86602540378, -.5) -- (1.86602540378+.5, -.5-0.86602540378); 
                \draw (1+2* 0.86602540378, 0) -- (1+2*0.86602540378+.5,-.86602540378); 
            \pgfmathsetmacro{\lx}{1.5 + 3* 0.86602540378} 
            \pgfmathsetmacro{\ly}{1.5 + 0.86602540378} 
            
            \begin{scope}[xshift = \lx cm, yshift = \ly cm]

            \draw (0,0) -- (0,1) -- (1,1) -- (1,0) -- (0,0);
            \begin{scope}[xshift = -0.86602540378cm, yshift = .5cm]
                \draw (30:1) -- (90:1) -- (150:1) -- (210:1) -- (270:1) -- (330:1) -- (30:1);
            \end{scope}
            \draw (-0.86602540378, 1.5) -- (-0.86602540378-.5, 1.5+0.86602540378); 
            \draw (-2* 0.86602540378, 1) -- (-2*0.86602540378-.5,1.86602540378); 
            \draw (-0.86602540378, -.5) -- (-0.86602540378-.5, -.5-0.86602540378); 
            \draw (-2* 0.86602540378, 0) -- (-2*0.86602540378-.5,-.86602540378); 
            \begin{scope}[xshift = 1.86602540378cm, yshift = .5cm]
                \draw (30:1) -- (90:1) -- (150:1) -- (210:1) -- (270:1) -- (330:1) -- (30:1);
            \end{scope}
                \draw (1.86602540378, 1.5) -- (1.86602540378+.5, 1.5+0.86602540378); 
                \draw (1+2* 0.86602540378, 1) -- (1+2*0.86602540378+.5,1.86602540378); 
                \draw (1.86602540378, -.5) -- (1.86602540378+.5, -.5-0.86602540378); 
                \draw (1+2* 0.86602540378, 0) -- (1+2*0.86602540378+.5,-.86602540378); 

            \end{scope}

        \end{scope}
        }
    }
\end{scope}
\end{scope}

\begin{scope}[xshift = 9cm, yshift = -7cm]
\draw (1.25,-.5) node {$(3^4.6)$};
 \clip (0,0) rectangle (2.5,2.5);
\begin{scope}[scale = .3, xshift = -0.05cm, yshift = -0.05cm]
  \foreach \i in {0,...,5}
    {
    \foreach \j in {-1,...,4}
        {
        \pgfmathsetmacro{\ii}{\i * 2.5 + \j * .5} 
        \pgfmathsetmacro{\jj}{\i * 0.86602540378 + \j * 3 * 0.86602540378} 
        \begin{scope}[xshift = \ii cm, yshift = \jj cm]
        \draw (0:1) -- (60:1) -- (120:1) -- (180:1) -- (240:1) -- (300:1) -- (0:1);
         \pgfmathsetmacro{\r}{2* 0.86602540378} 
        \draw   (0:1)   -- (30:\r) -- 
                (60:1)  -- (90:\r) --
                (120:1) -- (150:\r) --
                (180:1) -- (210:\r) --
                (240:1) -- (270:\r) --
                (300:1) -- (330:\r) -- (0:1);
        \draw  (0:2) -- (60:2) -- (120:2) -- (180:2) -- (240:2) -- (300:2) -- (0:2);
        \draw  (0:1)   -- (0:2);
        \draw  (60:1)  -- (60:2);
        \draw  (120:1) -- (120:2);
        \draw  (180:1) -- (180:2);
        \draw  (240:1) -- (240:2);
        \draw  (300:1) -- (300:2);
        \end{scope}
        }
    }
\end{scope}
\end{scope}

 \end{tikzpicture}

\caption{The $11$ Archimedean tilings}
\label{fig:11_tilings}
\end{figure}

There is a vast literature about various aspects of Archimedean
tilings. For historical details on Archimedean tilings we refer the
readers to \cite[Section 2.10]{GruenbaumS-89}. These tilings are relevant in
crystallography as layers of stacked $3$-dimensional structures
\cite{FK1,FK2}. Archimedean tiling structures at different length
scales have the potential to exhibit interesting properties: they may
form frustrated magnets \cite{Har04} or photonic crystals
\cite{UDG07}. Diffusion constants of Archimedean Tilings have been
calculated in \cite{Bas06}. Percolation thresholds of Archimedean
solids have been investigated, e.g., in \cite{SykesE-64, Kesten-80, Sud99, Jac14, Par07}.

We view these tilings as combinatorial graphs $G = (\V,\E)$ with
vertex set $\V$ and edge set $\E$.
The Laplacian $\Delta \colon \ell^2(\V) \to \ell^2(\V)$ on such a graph is defined as
\begin{equation} \label{eq:Lapl}
 (\Delta f)(v) 
 =
 f(v) - \frac{1}{\lvert v \rvert} \sum_{w \sim v} f(w),
\end{equation}
where $\lvert v \rvert$ denotes the vertex degree of $v \in \V$, and
$w \sim v$ means that $w$ and $v$ are adjacent, i.e. joined by an
edge.  This is a self-adjoint, bounded operator. On each of these
graphs, there is a cofinite $\ZZ^2$-action allowing to define the
Integrated Density of States (IDS) (for the precise definition see
Section~\ref{sec:2}).

Our first main results are concrete integral expressions of the
IDS for the Archimedean tilings $(4^4)$, $(3^6)$, $(6^3)$, $(3.6)^2$,
and $(3.12^2)$. Moreover, we show that the tilings $(3.6)^2$ (Kagome
lattice), and $(3.12^2)$ have $\Delta$-eigenfunctions of finite
support leading to jumps of the IDS.  Finally, we show that no other
Archimedean tiling has (any $\ell^2(\V)$) eigenfunctions.
\begin{remark}
 \label{rem:IDS_spectral_measure}
 For periodic graphs with co-finite $\ZZ^d$ action, the (distributional) derivative of the IDS, the \emph{density of states}, is a spectral measure in the sense that is carries all information on the spectrum:
 The points of increase of the IDS, i.e. the support of the density of states, are the spectrum of $\Delta$~\cite[p.119]{MathaiY-02}, see also \cite[Prop.~5.2]{LenzPV-07} for a proof of this statement in a more general context.
 The set of discontinuities of the IDS constitues the pure point spectrum and the singular continuous spectrum is empty~\cite[Theorem 6.10]{Kuchment-16}. 
 Thus, the remaining points of increase are the absolutely continuous spectrum.
 In particular, we have a complete description of the spectral types on all $11$ Archimedean lattices.
 
 Furthermore, since we have concrete expressions for the IDS of the tilings $(4^4)$, $(3^6)$, $(6^3)$, $(3.6)^2$,
and $(3.12^2)$, it is straightforward to calculate their densities of states from our expressions below.
\end{remark}

The method of proof is based on Floquet theory and can be applied to
more general graphs with cofinite $\ZZ^d$-action and not only to Archimedean tilings.
Examples include periodic finite hopping range operators on the nearest neighbour graph on $\ZZ^d$ or on non-planar, $\ZZ^2$-periodic graphs.

\section{General results on the IDS and the lattice $\ZZ^d$}

\label{sec:2}

\subsection{Floquet theory and the IDS}

Even though the goal of this article will be to study the $11$ (planar) graphs based on Archimedean tesselations, the results of this subsection do not require planarity of the graph.
More precisely, let $G = (\V,\E)$ be an infinite graph with vertex set $\V$ and edge
set $\E$.  We assume that the vertex degree $\lvert v \rvert$ is
finite for every $v \in \V$.

We also assume that there is a cofinite $\ZZ^d$-action on $G$,
given by 
\[
 \ZZ^d \ni \gamma 
 \mapsto
 T_\gamma : \V \to \V.
\]
Let $Q \subset \V$ be a (finite) fundamental domain of this action.

The graph Laplacian $\Delta$, a self-adjoint bounded operator on
$\ell^2(\V)$, was defined in \eqref{eq:Lapl}. The (abstract)
Integrated Density of States (IDS) $N_G : \RR \to [0,1]$ of the
Laplacian $\Delta$ on $G$ is
\[
 N_G(E)
 :=
 \frac{1}{\lvert Q \rvert}
 \tr ( \chi_Q \chi_{(- \infty, E]}(\Delta) )
\]
where $\chi_{(- \infty, E]}(\Delta)$ denotes the spectral projector onto the interval $(- \infty, E]$.
Intuitively, the IDS counts the number of states of $\Delta$ below the energy level $E$ per unit volume~\cite{LPPV09}. 
This is also reflected by Formula~\eqref{eq:integral_formula} below.
The IDS is non-decreasing and right continuous.

In order to apply Floquet theory, we also define the $d$-dimensional torus $\TT^d = \RR^d \slash (2 \pi \ZZ)^d$ and for every $\theta \in \TT^d$ the $\lvert Q \rvert$-dimensional Hilbert space
\[
 \ell^2(\V)_\theta
 :=
 \left\{
  \tilde f : \V \to \CC 
  \mid
  \tilde f(T_\gamma v) = e^{ i \langle \theta, \gamma \rangle} \tilde f(v)\ 
  \text{for all}\
  \gamma \in \ZZ^d
 \right\}
\]
with inner product
\[
 \langle f, g \rangle_\theta
 :=
 \sum_{v \in Q}
 f(v) \overline{g(v)}.
\]
Furthermore, we define on $\ell^2(\V)_\theta$ the $\theta$-pseudoperiodic Laplacian $\Delta^\theta$ as
\[
 \Delta^\theta f(v) 
 :=
 f(v) - \frac{1}{\lvert v \rvert}
 \sum_{w \sim v} f(w),
 \]
 that is, $\Delta^\theta$ acts in the same way as $\Delta$ but on the
 different vector space $\ell^2(\V)_\theta$. Since this is a
 $|Q|$-dimensional vector space due to quasiperiodicity, the operator
 $\Delta^\theta$ can be viewed as a hermitian $\lvert Q \rvert \times
 \rvert Q \rvert$-matrix.
 In Sections~\ref{sec:3} and~\ref{sec:4} we will give concrete examples of this matrix for the case of the $11$ Archimedean lattice graphs.
 The map $\TT^d \ni \theta \mapsto \sigma(\Delta^\theta)$ is also called \emph{dispersion relation}.

The following theorem provides an integral expression for the IDS on $\ZZ^d$-periodic graphs, see also~\cite[Theorem~6.18]{Kuchment-16}.

\begin{theorem}
 \label{thm:IDS_as_integral_expression}
\begin{equation}
  \label{eq:integral_formula}
 N_G(E)
 =
 \frac{1}{(2 \pi)^d \lvert Q \rvert}
 \int_{\TT^d} \# 
 \left\{
  \text{Eigenvalues of $\Delta^\theta$ less or equal than $E$}
 \right\}
 \mathrm{d} \theta
 .
\end{equation}
\end{theorem}


For the convenience of the reader we now give a proof of Theorem~\ref{thm:IDS_as_integral_expression} using Fourier theory on $\ell^2(\ZZ^d)$.

We have $\ell^2(\V) = \oplus_{v \in Q} \ell^2(\ZZ^d)$, where each summand
$\ell^2(\ZZ^d)$ represents the space $\ell^2(\{ T_\gamma v \mid \gamma \in
\ZZ^d \})$. Therefore, we can isometrically identify $f \in
\ell^2(\V)$ with $(f_v)_{v \in Q} \in \oplus_{v \in Q} \ell^2(\ZZ^d)$
by $f_v (\gamma) := f(T_\gamma v)$.  Applying the Fourier transform on
every component, we obtain
\[
 \hat f \in \oplus_{v \in Q} L^2(\TT^d),
 \quad
 \hat f := (\hat f_v)_{v \in Q},
 \quad
 \text{where}
 \quad
 \hat f_v(\theta)
 :=
 \sum_{\gamma \in \ZZ^d}
 \euler^{-i \langle \theta, \gamma \rangle} f_v(\gamma).
\]
From Fourier theory it follows that $f \mapsto \hat f$ is an isometry with the norms
\[
 \lVert f \rVert_{\ell^2(\V)}
 :=
 \sum_{v \in \V} \lvert f(v) \rvert^2
 =
 \sum_{v \in Q}
 \sum_{\gamma \in \ZZ^d}
\lvert (T_\gamma v) \rvert^2
\]
and
\[
 \lVert \hat f \rVert_{\oplus_{v \in Q} L^2(\TT^d)}
 :=
 \sum_{v \in Q} \lVert \hat f_v \rVert_{L^2(\TT^d)}^2
 \quad
 \text{where}
 \quad
 \lVert g \rVert_{L^2(\TT^d)}^2
 :=
 \frac{1}{(2 \pi)^d}
 \int_{\TT^d}
 \lvert g (\theta) \rvert^2 \drm \theta.
\]
We write $\tilde f_\theta(v) := \hat f_v(\theta)$ and extend $\tilde f_\theta(v)$ quasiperiodically to $\V$ via
\[
 \tilde f_\theta(T_\gamma v_0) 
 =
 \euler^{ i \langle \theta, \gamma \rangle} \tilde f_\theta(v_0) ,
 \quad
 \text{where $v_0 \in Q$}.
\]
We have isometrically identified the spaces
\[
 \ell^2(\V)
 \simeq
 \int_{\TT^d}^\oplus \ell^2(\V)_\theta \mathrm{d} \theta.
\]
\begin{lemma}
 For all $v \in \V$, all $f \in \ell^2(\V)$, and all $\theta \in \TT^d$, we have
 \begin{equation}
  \label{eq:lemma_1}
  \tilde f_\theta(v)
  =
  \sum_{\gamma \in \ZZ^d} \euler^{- i \langle \theta, \gamma \rangle} f(T_\gamma v).
 \end{equation}
\end{lemma}

\begin{proof}
 Write $v = T_{\gamma_0} v_0$ for $v_0 \in Q$.
 Then
\begin{align*}
 \tilde f_\theta(v)
 =
 \tilde f_\theta(T_{\gamma_0} v_0)
 &=
 \euler^{ i \langle \theta, \gamma_0 \rangle} 
 \tilde f_\theta(v_0)
 =
 \euler^{ i \langle \theta, \gamma_0 \rangle} 
 \hat f_{v_0}(\theta)
 =
 \sum_{\gamma \in \ZZ^d}
 \euler^{ i \langle \theta, \gamma_0 \rangle} 
 \euler^{-i \langle \theta, \gamma \rangle} 
 f(T_\gamma v_0)\\
 &=
 \sum_{\gamma \in \ZZ^d}
 \euler^{- i \langle \theta, \gamma - \gamma_0 \rangle} 
 f(T_{\gamma - \gamma_0} T_{\gamma_0} v_0)
 =
 \sum_{\gamma' \in \ZZ^d}
 \euler^{- i \langle \theta, \gamma' \rangle} 
 f(T_{\gamma'} v).
 \qedhere
\end{align*}

\end{proof}

Now, we can identify operators:
\begin{align}
 \label{eq:identify_operators}
 \Delta 
 \quad
 \cong
 \quad
 \int_{\theta \in \TT^d}^\oplus \Delta^\theta \mathrm{d} \theta
\end{align}
where
\[
 \Delta f(v) 
 =
 f(v) - \frac{1}{\lvert v \rvert} \sum_{w \sim v} f(w)
\]
and $\Delta f = g$, if and only if for all $v = T_{\gamma_0} v_0$ with $v_0 \in Q$
\begin{align*}
 \tilde g_\theta(v)
 & \stackrel{\text{\eqref{eq:lemma_1}}}{=}
 \sum_{\gamma \in \ZZ^d} e^{- i \langle \theta, \gamma \rangle} g(T_\gamma v)
 = 
 \sum_{\gamma \in \ZZ^d} e^{- i \langle \theta, \gamma \rangle} \left[ f(T_\gamma v) - \frac{1}{\lvert T_\gamma v \rvert} \sum_{w \sim T_\gamma v} f(w) \right]\\
 &  \stackrel{\text{\eqref{eq:lemma_1}}}{=}
 \tilde f_\theta(v)
 -
 \frac{1}{\lvert v \rvert}
 \sum_{\gamma \in \ZZ^d} \sum_{w' \sim v} e^{-i \langle \theta, \gamma \rangle} f(T_\gamma w')
  =
 \tilde f_\theta(v)
 -
 \frac{1}{\lvert v \rvert}
 \sum_{w \sim v} \sum_{\gamma \in \ZZ^d} e^{- i \langle \theta, \gamma \rangle} f(T_\gamma w)\\ 
 & \stackrel{\text{\eqref{eq:lemma_1}}}{=}
 \tilde f_\theta(v)
 -
 \sum_{w \sim v} 
 \frac{1}{\lvert v \rvert}
 \tilde f_\theta(w) 
 =
 \Delta^\theta \tilde f_\theta (v).
\end{align*}
Recall that $\Delta^\theta$ and $\Delta$ are formally defined via the same expressions, but they operate on different spaces: $\Delta$ operates on $\ell^2$-functions on $G$ while $\Delta^\theta$ operates on $\theta$-quasiperiodic functions. 
\par
From~\eqref{eq:identify_operators}, we conclude
\[
 \chi_{(- \infty, E]}(\Delta)
 \cong
 \int_{\TT^d}^\oplus  
 \chi_{(- \infty, E]}(\Delta^\theta)
 \drm \theta.
\]
and therefore
\[
 (\widetilde{ \chi_{(- \infty, E]}(\Delta) f})_\theta
 =
 \chi_{(- \infty, E]}(\Delta^\theta) \widetilde f_\theta
\]
Now, we are in a position to calculate the IDS. We have
\begin{align*}
 N_G(E)
 &=
 \frac{1}{\lvert Q \rvert} \operatorname{Tr} \left( \chi_Q \chi_{(- \infty, E]}(\Delta) \right)
 =
 \frac{1}{\lvert Q \rvert} 
 \sum_{v \in Q} 
 \langle \delta_v, \chi_{(- \infty, E]}(\Delta) \delta_v \rangle\\
 &= 
 \frac{1}{(2 \pi)^d \lvert Q \rvert} 
 \sum_{v \in Q} 
 \int_{\TT^d} 
 \langle ( \widetilde{\delta_v})_\theta , ( \widetilde{\chi_{(- \infty, E]}(\Delta^\theta) \delta_v})_\theta \rangle \mathrm{d} \theta\\
 &=
 \frac{1}{(2 \pi)^d \lvert Q \rvert} 
 \int_{\TT^d} 
  \sum_{v \in Q}
  \langle (\widetilde{\delta_v})_\theta , \chi_{(- \infty, E]}(\Delta^\theta) \widetilde{\delta_v})_\theta  \rangle \mathrm{d} \theta
\end{align*}
The operator, $\chi_{(- \infty, E]}(\Delta^\theta)$ is an orthogonal projection onto the finite-dimensional span of eigenfunctions of $\Delta^\theta$ on $\ell^2(Q)$ with eigenvalues smaller or equal than $E$ (i.e. a matrix).
Hence, the trace $\operatorname{Tr} \left( \chi_{(-\infty, E]}(\Delta^\theta) \right)$ is the number of eigenvalues of $\Delta^\theta$ less or equal than $E$.
This finishes the proof of Theorem~\ref{thm:IDS_as_integral_expression}.\qed
\par
\medskip
The next results are useful to show absence of finitely supported eigenfunctions for particular graphs.

\begin{theorem}
 \label{thm:equivalence}
 The following are equivalent:
 \begin{enumerate}[(i)]
  \item 
  $N_G$ is continuous at $E$.
  \item
  $\Delta$ has no eigenfunctions with eigenvalue $E$ of finite support.
  \item
  $\Delta$ has no $\ell^2(\V)$-eigenfunctions with eigenvalue $E$.
  \item
  There is $\theta \in \TT^d$ such that $E \not\in \sigma(\Delta^\theta)$.
 \end{enumerate}
\end{theorem}

\begin{corollary}
 \label{cor:IDS_continuous}
 If there exist $\theta, \theta' \in \TT^d$ such that $\sigma(\Delta^{\theta}) \cap \sigma(\Delta^{\theta'}) = \emptyset$, then $N_G$ is continuous.
\end{corollary}

\begin{proof}[Proof of Theorem~\ref{thm:equivalence}]
 The equivalence of items (i), (ii) and (iii) is proved in~\cite{Kuchment-91}, see also~\cite[Corollary~2.3]{LenzV-09} for a proof in a more general setting.
 
 It remains to show the equivalence of (i) and (iv). 
 We fix $E \in \RR$ and calculate, using the dominated convergence theorem,
 \begin{align*}
  N_G(E) &- \lim_{E' \nearrow E} N_G(E')
  =
  \frac{1}{(2 \pi)^d \lvert Q \rvert}
  \int_{\TT^d} 
  \lim_{E' \nearrow E} 
  \#
  \left\{
  \text{Eigenvalues of $\Delta^\theta$ in $(E',E]$}
  \right\}
  \mathrm{d} \theta
  \\
  &=
  \frac{1}{(2 \pi)^d \lvert Q \rvert}
  \int_{\TT^d} 
%
  \left\{
  \text{Multiplicity of the eigenvalue $E$ of $\Delta^\theta$}
  \right\}
  \mathrm{d} \theta.
 \end{align*}
 This is non-zero if and only if the characteristic polynomial
 \[
  P_{\Delta^\theta}(E) 
  :=
  \det
  \left( \Delta^\theta - E \cdot \id \right)
 \]
 vanishes on a set $S \subset \TT^d$ of positive measure.
 Since $\theta \mapsto P_{\Delta^\theta}(E)$ is a real analytic function, this is equivalent to $P_{\Delta^\theta}(E)$ vanishing identically on $\TT^d$ (see~\cite[p.~67]{KrantzP-92}).
 Thus discontinuity of $N_G$ at $E$ is equivalent to $E \in \sigma(\Delta^\theta)$ for all $\theta \in \TT^d$.
\end{proof}
Let us note that the analytic nature of the band functions has been used in similar arguments before, see e.g.~\cite[Corollary~6.19]{Kuchment-16}.

\subsection{The lattice $\ZZ^d$}

As a first application of~\eqref{eq:integral_formula}, we calculate the IDS of $\Delta$ on the lattice $\ZZ^d$.
An elementary cell $Q$ consists of a single point. In the $2$-dimensional case, we can view $\ZZ^2$ as a tiling by unit squares (i.e. as the $(4^4)$ tiling)
with $\ZZ^2$ generated by translation vectors $\omega_1 = (1,0)$, and $\omega_2 = (0,1)$, cf.~ Figure~\ref{fig:ZZ^2}.
The $(1 \times 1)$-matrix corresponding to $\Delta^\theta$ has the entry (and hence the only eigenvalue)
\begin{equation}
 \label{eq:lambda_ZZ^d}
 \lambda_{\ZZ^d}^\theta
 =
  1 
  -
  \frac{1}{2d}
  (
  \euler^{- i \theta_1} + \euler^{i \theta_1}
  +
  \dots
  +
  \euler^{- i \theta_d} + \euler^{i \theta_d}
  )
 =
 1 - \frac{1}{d} \sum_{j = 1}^d \cos(\theta_j).
\end{equation}

Thus,~\eqref{eq:integral_formula} simplifies to
\begin{equation}
 \label{eq:IDS_ZZ^d}
 N_{\ZZ^d}(E)
 =
 \frac{1}{(2 \pi)^d}
 \operatorname{Vol} 
 \left\{
  \theta \in \TT^d \colon \frac{1}{d} \sum_{j = 1}^d \cos(\theta_j) \geq 1 - E
 \right\}.
\end{equation}

It is clear that $N_{\ZZ^d}(E)$ is supported in $[0,2]$. 
Moreover, by Corollary~\ref{cor:IDS_continuous}, the IDS on $\ZZ^d$ is continuous and $\Delta$ has no $\ell^2$-eigenfunctions since from~\eqref{eq:lambda_ZZ^d} we conclude $\lambda_{\ZZ^d}^\theta = 0 \neq 2 = \lambda_{\ZZ^d}^{\theta'}$ for $\theta = (0,\dots, 0)$, and $\theta' = (\pi, \dots, \pi)$.

In dimensions $d = 1, 2$, the following expressions for the IDS follow
directly from~\eqref{eq:IDS_ZZ^d}. In the case $d=2$, we derive the
expression by applying the substitution $t =\cos \theta_1$.

\begin{proposition}
In dimension $d = 1$, we have
\[
 N_{\ZZ}(E)
 =
 \frac{\chi_{[0,2]}(E)}{\pi}
 \cdot 
 \arccos(1 - E).
\]
In dimension $d = 2$, we have
\begin{eqnarray*}
 N_{\ZZ^2}(E) = N_{(4^4)}(E) &=& \frac{1}{(2\pi)^2} \operatorname{Vol} 
 \left\{ \theta \in \TT^2 \mid \cos \theta_1 + \cos \theta_2 \ge 2-2E \right\} \\
 &=&
 \begin{cases}
  0
   &\quad
   \text{if}\ E < 0,\\
  \frac{1}{\pi^2}\int_{1 - 2 E}^1 \frac{\arccos(2 - 2 E - t)}{\sqrt{1-t^2}} \mathrm{d} t
   &\quad
   \text{if}\ 0 \leq E \leq 1,\\ 
  1 - N_{(4^4)}(2 - E)
   &\quad
   \text{if}\ 1 < E \leq 2,\\
  1
   &\quad
   \text{if}\ 2 < E.
 \end{cases}
\end{eqnarray*}
\end{proposition}
\begin{figure}[ht]
 \begin{subfigure}{.45\textwidth}
 \begin{tikzpicture}
  \draw[fill = black] (0,0) circle (2pt);
  \draw[fill = white] (0,1) circle (2pt);
  \draw[fill = white] (0,-1) circle (2pt);
  \draw[fill = white] (1,0) circle (2pt);
  \draw[fill = white] (-1,0) circle (2pt);  
  
  \draw[thick] (0,.2) -- (0,.8);
  \draw[thick] (0,-.2) -- (0,-.8);
  \draw[thick] (.2,0) -- (.8,0);
  \draw[thick] (-.2,0) -- (-.8,0);

  \draw (.2,.2) node {$a$};
  \draw[fill = white] (1.3,.2) node {$(a + \omega_1)$};
  \draw[fill = white] (-1.3,.2) node {$(a - \omega_1)$};
  \draw[fill = white] (0,1.2) node {$(a + \omega_2)$};
  \draw[fill = white] (0,-1.2) node {$(a - \omega_2)$};  
 \end{tikzpicture}
\end{subfigure}
\begin{subfigure}{.45\textwidth}
 \begin{tikzpicture}
  \clip ({-.5*\linewidth},-2.5) rectangle ({.5*\linewidth},2.5);
  \draw (0,0) node {\includegraphics[width=\linewidth]{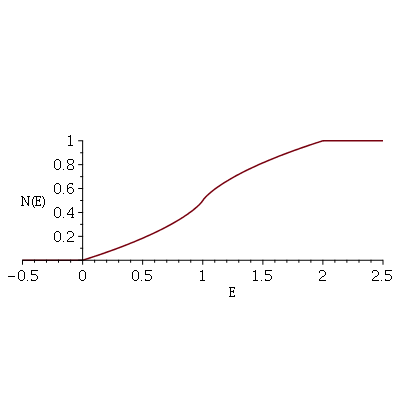}};
 \end{tikzpicture}
\end{subfigure}
\caption{Fundamental domain of the graph $\ZZ^2$ (left) and its IDS (right)}
\label{fig:ZZ^2}
\end{figure}
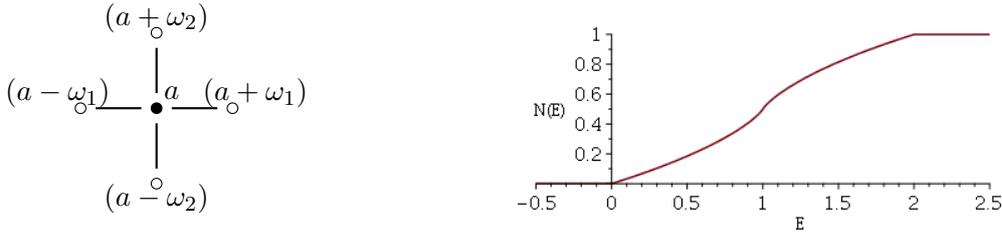


\section{Concrete integral expressions for the IDS of some Archimedean tilings}
\label{sec:3}

In this section, we present concrete integral expressions of the IDS of the Archimedean tilings with vertex types $(3^6)$, $(6^3)$, $(3.6)^2$, and $(3.12^2)$.
We will denote the corresponding IDS by $N_{(3^6)}$, etc.

We will see that only the last two tilings admit finitely supported eigenfunctions.

\subsection{IDS of the $(3^6)$ tiling (triangular lattice)}

A fundamental domain consists of a single point with translation vectors $\omega_1 = (1,0)$, $\omega_2 = (\cos(\pi/3),\sin(\pi/3))$, cf. Figure~\ref{fig:(3^6)}. 
The corresponding matrix $\Delta^\theta$ has the only entry and hence the only eigenvalue 
\begin{align*}
 \lambda_{(3^6)}^\theta
 &=
 \left( 1 - \frac{1}{6} (e^{i \theta_1} + e^{- i \theta_1} + e^{i \theta_2} + e^{- i \theta_2} + e^{i (\theta_2 - \theta_1)} + e^{- i (\theta_2 - \theta_1)}) \right)\\
 &=
 \left( 1 - \frac{1}{3} (\cos(\theta_1) + \cos(\theta_2) + \cos(\theta_2 - \theta_1) ) \right).
\end{align*}
Therefore,
\begin{equation}
 \label{eq:volume_formula_(3^6)}
 N_{(3^6)}(E)
 =
 \frac{1}{(2 \pi)^2}
 \operatorname{Vol}
 \left\{
  \theta \in \TT^2 
  \colon
  \frac{1}{3} (\cos(\theta_1) + \cos(\theta_2) + \cos(\theta_2 - \theta_1)) \geq 1 - E
 \right\}
\end{equation}

\begin{figure}[ht]

\begin{subfigure}{.45\textwidth}
 \begin{tikzpicture}
  \draw[fill = black] (0,0) circle (2pt);
  \draw[fill = white] (0:1) circle (2pt);
  \draw[fill = white] (60:1) circle (2pt);
  \draw[fill = white] (120:1) circle (2pt);  
  \draw[fill = white] (180:1) circle (2pt);
  \draw[fill = white] (240:1) circle (2pt);
  \draw[fill = white] (300:1) circle (2pt);
  \draw[thick] (0:.8)   -- (0:.2);
  \draw[thick] (60:.8)  -- (60:.2);
  \draw[thick] (120:.8) -- (120:.2);  
  \draw[thick] (180:.8) -- (180:.2);
  \draw[thick] (240:.8) -- (240:.2);
  \draw[thick] (300:.8) -- (300:.2);

  \draw (0,.35) node {$a$};
  \draw (0:1.9) node {$(a + \omega_1)$};
  \draw (55:1.8) node {$(a + \omega_2)$};
  \draw (125:1.8) node {$(a - \omega_1 + \omega_2)$};  
  \draw (180:1.9) node {$(a - \omega_1)$};
  \draw (235:1.8) node {$(a - \omega_2)$};
  \draw (305:1.8) node {$(a + \omega_1 - \omega_2)$};

\end{tikzpicture}
\end{subfigure}
%
%
\begin{subfigure}{.45\textwidth}
 \begin{tikzpicture}
  \clip ({-.5*\linewidth},-2.5) rectangle ({.5*\linewidth},2.5);
  \draw (0,0) node {\includegraphics[width=\linewidth]{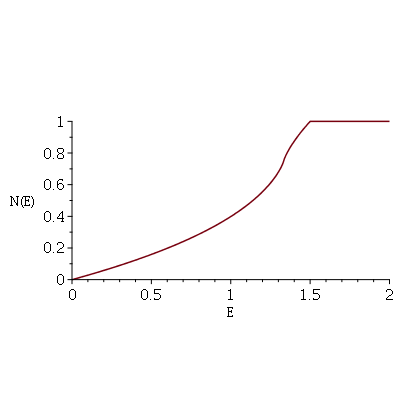}};
 \end{tikzpicture}
\end{subfigure}

\caption{Fundamental domain of the $(3^6)$ tiling (left) and its IDS (right)}
\label{fig:(3^6)}
\end{figure}
Since the expression 
\[
 \operatorname{Vol}
 \left\{
 \theta \in \TT^2
 \colon
 \cos(\theta_1) + \cos(\theta_2)) + \cos(\theta_2 - \theta_1) 
 \geq 
 L
 \right\}
\]
will be relevant later on, we shall discuss it in more detail here.
By periodicity, we can consider $\TT^2$ as $(-\pi, \pi)^2$ (the boundary is a measure zero set and does not play any role).
Using the change of variables $u := (\theta_1 + \theta_2)/\sqrt{2}$, $v = (\theta_1 - \theta_2)/\sqrt{2}$ we find
\begin{align*}
 \cos (\theta_1) + \cos (\theta_2) + \cos (\theta_1 - \theta_2)
 &=
 \cos \left( \frac{u + v}{\sqrt{2}} \right)
 +
 \cos \left( \frac{u - v}{\sqrt{2}} \right)
 +
 \cos \left( \sqrt{2} u \right)\\
 &=
 2\cos \left( \frac{u}{\sqrt{2}} \right)
 \cos \left( \frac{v}{\sqrt{2}} \right)
 +
 2 \cos^2 \left( \frac{u}{\sqrt{2}} \right)
 - 
 1.
\end{align*}
The new variables $(u,v)$ identify $\TT^2$ with the domain $\largediamond := \{ (u,v) \in \RR^2 \colon \lvert u \rvert + \lvert v \rvert < \pi \sqrt{2} \}$.
 
  \begin{lemma}
 \label{lem:F}
 The function $F : \largediamond \to \RR$, defined by 
 \[
 F(u,v) = 2\cos \left( \frac{u}{\sqrt{2}} \right)
 \cos \left( \frac{v}{\sqrt{2}} \right)
 +
 2 \cos^2 \left( \frac{u}{\sqrt{2}} \right)
 -
 1
 \]
 has the following properties:
 \begin{enumerate}[i)]
   \item The global maximum of $F$ is at $(u,v) = (0,0)$, where $F(u,v) = 3$.
   \item The two global minima of $F$ are at $(u,v) = (\pm 2/3 \cdot \sqrt{2} \pi,0)$, where $F(u,v) = - 3 /2$.
   \item $F \geq -1$ in the ``hexagon'' $\hex := \{ (u,v) \in \largediamond \colon \lvert u \rvert \leq \pi / \sqrt{2} \}$ and $F \leq -1$ in the complemetary set $\tri_- \cup \tri_+$ which consists of two rectangular triangles.
   \item We have
   \begin{align*}
    \operatorname{Vol} \{ F \geq L \}
    =
    \begin{cases}
     (2 \pi)^2
      &\quad
      \text{if}\ L < -3/2,\\
      (2 \pi^2) - 8
      \int_{ - 1 /2 -\frac{1}{2} \sqrt{ 2 L + 3}}^{ - 1 /2 +\frac{1}{2} \sqrt{ 2 L + 3}}
      \frac{\arccos \left( \frac{L + 1}{2 t} - t \right)}{\sqrt{1 - t^2}}
      \mathrm{d}t
      &\quad
      \text{if}\ -3/2 \leq L < -1,\\
      8 \int_{- \frac{1}{2} + \frac{1}{2} \sqrt{2 L + 3}}^1
      \frac{\arccos \left( \frac{L + 1}{2 t} - t \right)}{\sqrt{1 - t^2}}
      \mathrm{d}t
      &\quad
      \text{if}\ -1 \leq L < 3,\\
     0
      &\quad
      \text{if}\ 3 \leq L.
    \end{cases}
   \end{align*}
  \end{enumerate} 
\end{lemma}

\begin{remark}
 It is known that additional symmetries (e.g. a rotational symmetries) of the underlying graph are reflected in symmetries of the dispersion relation.
 More precisely, in an appropriate basis, the function $F$ is symmetric under rotations by $\pi/3$ around its maximum and symmetric under rotations by $2 \pi/3$ around its minima.
 This corresponds to symmetries of the underlying graph, see~\cite[Lemma~2.1]{BerkolaikoC-18} for details.
\end{remark}

\begin{proof}
 It is straightforward to check i) to iii) and using symmetry and monotonicity considerations
  \begin{align*}
    \operatorname{Vol} \{ F \geq L \}
    =
    \begin{cases}
     (2 \pi)^2
      &\quad
      \text{if}\ L < -3/2,\\
     (2 \pi)^2 - 2 \cdot \vol \left\{ (u,v) \in \tri_+ \colon  F (u,v) \geq L \right\}
      &\quad
      \text{if}\ -3/2 \leq L < -1,\\
     4 \cdot \vol \left\{ (u,v) \in \hex \cap \RR_+^2 \colon F(u,v) \geq L \right\}
      &\quad
      \text{if}\ -1 \leq L < 3,\\
     0
     &\quad
     \text{if}\ 3 \leq L.
    \end{cases}
   \end{align*}
%
%
\begin{figure}
\begin{tikzpicture}
 \draw (0,0) node {\includegraphics[scale=.4]{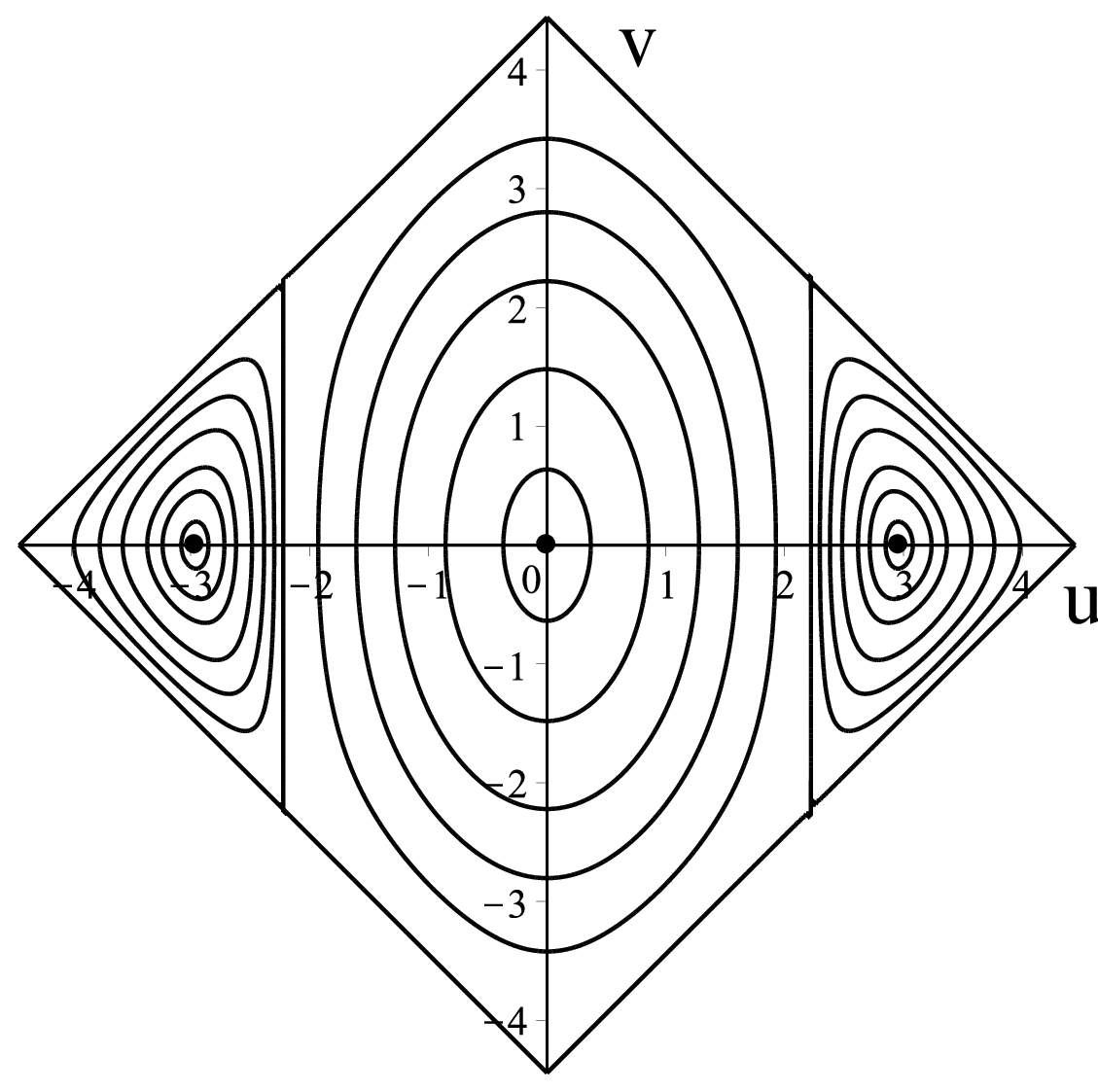}};
 \draw (-3.5,1) node {$\tri_-$};
 \draw (3.5,1) node {$\tri_+$};
 \draw (-1.75,3) node {$\hex$};
\end{tikzpicture}
\caption{Level sets of the function $F \colon \largediamond \to \RR$ and the domains $\hex$ and $\tri_{\pm}$}
 \end{figure}

 To calculate the area within $\tri_+$ we consider the upper half (i.e. $v \geq 0$) of $\tri_+$ (i.e. $u \geq \pi/\sqrt{2}$).
 Therein, $\cos( u/\sqrt{2}) < 0$, whence $F(u,v) \leq L$ is equivalent to
 \begin{equation}
 \label{eq:triangle_1}
  \cos \left( \frac{v}{\sqrt{2}} \right)
  \geq
  \frac{L + 1}{2 \cos \left( \frac{u}{\sqrt{2}} \right)} 
  - 
  \cos \left( \frac{u}{\sqrt{2}} \right).
 \end{equation}
 and we found that the area is the area under a graph.
 Since $\cos( v/ \sqrt{2} ) \leq 1$, we conclude that~\eqref{eq:triangle_1} can only be fulfilled if $u$ is in the interval between the two solutions of $\cos(u/\sqrt{2}) = (L + 1)/2 - \cos^2(u/\sqrt{2})$ in $(\pi/\sqrt{2}, \sqrt{2} \pi)$, i.e. for
 \[
  u 
  \in
  (u_-,u_+)
  :=
  \left(
   \sqrt{2} 
   \arccos
    \left( 
     - 1 /2 +\frac{1}{2} \sqrt{ 2 L + 3}
    \right),
   \sqrt{2} 
   \arccos
    \left( 
     - 1 /2 -\frac{1}{2} \sqrt{ 2 L + 3}
     \right)
  \right).
 \]
 Together with~\eqref{eq:triangle_1}, we find
 \begin{align*}
  &\vol \left\{ (u,v) \in \tri_+ \colon  F (u,v) \geq L \right\} \\
  =
  &2 \int_{u_-}^{u_+}
   \sqrt{2} 
   \arccos
    \left(
     \frac{L + 1}{2 \cos \left( \frac{u}{\sqrt{2}} \right)} 
      - 
      \cos \left( \frac{u}{\sqrt{2}} \right)
    \right)
    \mathrm{d} u \\
   =
   &\int_{ - 1 /2 -\frac{1}{2} \sqrt{ 2 L + 3}}^{ - 1 /2 +\frac{1}{2} \sqrt{ 2 L + 3}}
    4 \frac{ 
    \arccos 
     \left(
      \frac{L + 1}{2 t} - t
     \right)
    }
    {\sqrt{1 - t^2}}
    \mathrm{d}t
 \end{align*}
 where in the last step, we used the transformation $u = \sqrt{2}\arccos(t)$.
 \par
 As for the area in the hexagon $H$, by an analogous argument,
 \begin{align*}
  &\vol \left\{ (u,v) \in \hex \cap \RR_+^2 \colon F(u,v) \geq L \right\} =\\
  =
  &\int_0^{u_0} 
  \sqrt{2} 
  \arccos
   \left(
    \frac{L + 1}{2 \cos (u / \sqrt{2})}
    -
    \cos ( u/\sqrt{2})
   \right)
   \mathrm{d}u
  =
  2 
  \int_{- \frac{1}{2} + \frac{1}{2} \sqrt{2 L + 3}}^1
  \frac{ 
    \arccos 
     \left(
      \frac{L + 1}{2 t} - t
     \right)
    }
    {\sqrt{1 - t^2}}
    \mathrm{d}t
 \end{align*}
 where 
 \[
  u_0 
  = 
  \sqrt{2} \arccos 
   \left(
    - \frac{1}{2} + \frac{1}{2} \sqrt{2 L + 3} 
   \right)
 \]
 is the solution of $\cos(u/\sqrt{2}) = (L + 1)/2 - \cos^2(u/\sqrt{2})$ in $[0, \pi/\sqrt{2}]$.
\end{proof}

Combining Lemma~\ref{lem:F} and~\eqref{eq:volume_formula_(3^6)}, we find:

\begin{proposition}
\begin{align*}
 N_{(3^6)}(E)
 &=
 \frac{1}{(2 \pi)^2}
 \operatorname{Vol}
 \left\{
  (u,v) \in \largediamond
  \colon
  F(u,v)
  \geq
  3 - 3 E
 \right\}
 \\
 &=
 \begin{cases}
  0
   &\quad 
   \text{if}\ E < 0,\\
   \frac{2}{\pi^2}
   \int_{- \frac{1}{2} + \frac{1}{2} \sqrt{9 - 6 E}}^1
   \frac{ 
    \arccos 
     \left(
      \frac{4 - 3 E}{2 t} - t
     \right)
    }
    {\sqrt{1 - t^2}}
    \mathrm{d}t
   &\quad 
   \text{if}\ 0 \leq E < \frac{4}{3},\\
   1 
   -
   \frac{2}{\pi^2}
   \int_{ - 1 /2 -\frac{1}{2} \sqrt{9 - 6 E}}^{ - 1 /2 +\frac{1}{2} \sqrt{9 - 6 E}}
    \frac{ 
    \arccos 
     \left(
      \frac{4 - 3 E}{2 t} - t
     \right)
    }
    {\sqrt{1 - t^2}}
    \mathrm{d}t
   &\quad 
   \text{if}\ \frac{4}{3} \leq E < \frac{3}{2},\\
   1
    &\quad 
    \text{if}\ 3/2 < E.\\
 \end{cases} 
\end{align*}
In particular, $N_{(3^6)}$ is continuous and there are no $\ell^2$-eigenfunctions.
\end{proposition}

\subsection{IDS of the $(6^3)$ (honeycomb) tiling}

The honeycomb tiling is of particular practical interest since this
structure appears in graphene and is closely related to fullerenes
(buckeyballs) and carbon nano-tubes. 
The earliest reference from which the dispersion relations for this tiling can be inferred seems to be~\cite{Wallace-46}.
Furthermore, parts of our calculations have an overlap with the metric graph investigations in
\cite{KP07},
where the authors derive dispersion relations and
determine various spectral types of the Hamiltonian not only for the
$(6^3)$ tiling, but also for metric nano-tube graphs isometrically
embedded in cylinders. 
Moreover, \cite{DM10} is a good source to find
further information and references about graphene under a magnetic
field.

\begin{figure}[ht]
\begin{subfigure}{.45\textwidth}
 \begin{tikzpicture}
  \draw[fill = black] (0,0) circle (2pt);
  \draw[fill = black] (0,1) circle (2pt);
  \draw[fill = white] (-30:1) circle (2pt);
  \draw[fill = white] (-150:1) circle (2pt);
  \begin{scope}[yshift = 1cm]
  \draw[fill = white] (30:1) circle (2pt);  
  \draw[fill = white] (150:1) circle (2pt);   
  \end{scope}

  \draw[thick] (0,.2) -- (0,.8);
  \draw[thick] (-30:.8)   -- (-30:.2);
  \draw[thick] (-150:.8)   -- (-150:.2);
  \begin{scope}[yshift = 1cm]
  \draw[thick] (30:.8)   -- (30:.2);
  \draw[thick] (150:.8)   -- (150:.2);   
  \end{scope}

  \draw (.2,.2) node {$b$};
  \draw (.2,.8) node {$a$};
  \draw (-30:1.5) node {$(a - \omega_2)$};
  \draw (-150:1.5) node {$(a - \omega_1)$};
  \begin{scope}[yshift = 1cm]
  \draw (30:1.5) node {$(b + \omega_1)$};
  \draw (150:1.5) node {$(b + \omega_2)$};
  \end{scope}
\end{tikzpicture}  
\end{subfigure}
\begin{subfigure}{.5\textwidth}
 \begin{tikzpicture}
  \clip ({-.5*\linewidth},-2.5) rectangle ({.5*\linewidth},2.5);
  \draw (0,0) node {\includegraphics[width=\linewidth]{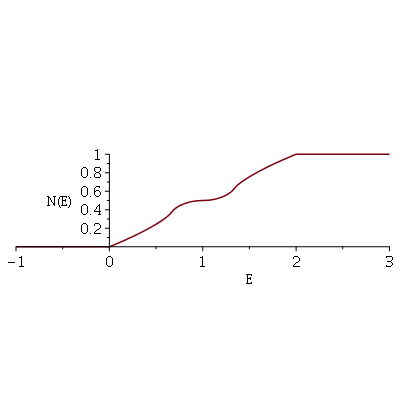}};
 \end{tikzpicture}
\end{subfigure}

\caption{Fundamental domain of the $(6^3)$ tiling (left) and its IDS (right)}
\label{fig:(6^3)}
\end{figure}
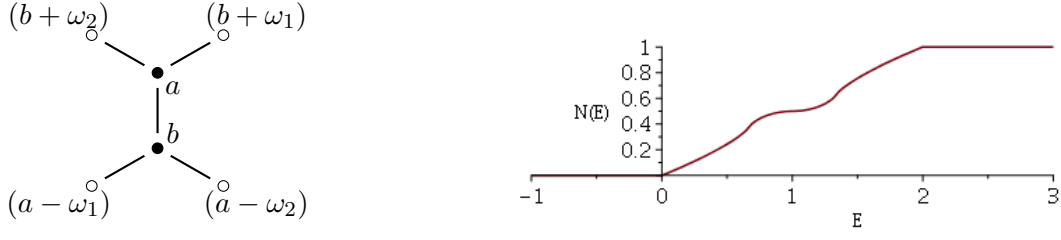

A fundamental domain is $Q = \{ a,b \} = \{ (0,0), (0,1) \}$, cf. Figure~\ref{fig:(6^3)}.
This implies
\[
 \Delta^\theta 
 = 
  \begin{pmatrix}
  1 & - \frac{1}{3} ( 1 + e^{i \theta_1} + e^{i \theta_2} ) \\
  - \frac{1}{3} ( 1 + e^{-i \theta_1} + e^{-i \theta_2} )  & 1\\  
 \end{pmatrix}
\]
which has the eigenvalues
\[
 \begin{cases}
  \lambda_{(6^3),1}^\theta &=
  1
  -
  \frac{1}{3}\sqrt
  {
   2 \cos \theta_1 + 2 \cos \theta_2 + 2 \cos(\theta_1 - \theta_2) + 3
  } 
  \\
  \lambda_{(6^3),2}^\theta &=
  1
  +
  \frac{1}{3}\sqrt
  {
   2 \cos \theta_1 + 2 \cos \theta_2 + 2 \cos(\theta_1 - \theta_2) + 3.
  } 
 \end{cases}
\]
Therefore,
\begin{align*}
 N_{(6^3)}(E)
 &=
 \frac{1}{2(2 \pi)^2}
 \left(
 \operatorname{Vol}
  \left\{
   \theta \in \TT^2 
   \colon
   \lambda_{(6^3),1}^\theta \leq E
  \right\}
 +
 \operatorname{Vol}
  \left\{
   \theta \in \TT^2 
   \colon
   \lambda_{(6^3),2}^\theta \leq E,
  \right\}
 \right)
\end{align*}
We see that $N_{(6^3)}(E)$ has support $[0,2]$ and is antisymmetric around $(E,N(E)) = (1,1/2)$.
For $E < 1$, we find by Lemma~\ref{lem:F}
\begin{align*}
 N_{(6^3)}(E)
 &=
 \frac{1}{2(2 \pi)^2}
 \operatorname{Vol}
  \left\{
   \theta \in \TT^2 
   \colon
   \lambda_{(6^3),1}^\theta \leq E
  \right\}\\
 &=
 \frac{1}{2(2 \pi)^2}
  \operatorname{Vol}
  \left\{
   (u,v) \in \largediamond 
   \colon
   F(u,v) \geq \frac{9}{2} (1 - E)^2 - 3/2
  \right\}\\
\end{align*}
 Therefore, using Lemma~\ref{lem:F} and antisymmetry around $E =1$, we find:
 
\begin{proposition}
 
\begin{align*}
 N_{(6^3)}(E)
 =
 \begin{cases}
  0
   &\quad
   \text{if}\ E < 0\\
  \frac{1}{\pi^2} 
   \int_{1 - \frac{3}{2} E}^1 
   \frac{\arccos \left( \frac{9(1-E)^2 - 1}{4t} - t \right)}{\sqrt{1 - t^2}}
   \mathrm{d}t
   &\quad
   \text{if}\ 0 \leq E < 2/3\\
  \frac{1}{2}
   -
   \frac{1}{\pi^2} 
   \int_{-2 + \frac{3}{2} E}^{1 - \frac{3}{2} E}  
   \frac{\arccos \left( \frac{9(1-E)^2 - 1}{4t} - t \right)}{\sqrt{1 - t^2}}
   \mathrm{d}t
    &\quad
   \text{if}\ 2/3 \leq E < 1\\
  \frac{1}{2}
   +
   \frac{1}{\pi^2} 
   \int_{1 - \frac{3}{2} E}^{-2 + \frac{3}{2} E}
   \frac{\arccos \left( \frac{9(E-1)^2 - 1}{4t} - t \right)}{\sqrt{1 - t^2}}
   \mathrm{d}t
    &\quad
   \text{if}\ 1 \leq E < 4/3\\
  1 
   -
   \frac{1}{\pi^2} 
   \int_{1 - \frac{3}{2} E}^1 
   \frac{\arccos \left( \frac{9(E-1)^2 - 1}{4t} - t \right)}{\sqrt{1 - t^2}}
   \mathrm{d}t
   &\quad
   \text{if}\ 4/3 \leq E < 2\\
  1
   &\quad
   \text{if}\ 2 \leq E.
 \end{cases}
\end{align*}
In particular, there are no $\ell^2$-eigenfunctions.
\end{proposition}

\subsection{IDS of the $(3.6)^2$ tiling (``the Kagome lattice'')}

Properties of Kagome lattice structures under magnetic fields have
been investigated both in the Applied Physics literature (e.g.,
\cite{Zhao16} and references therein studying Kagome staircase
compounds) and in the Theoretical Physics literature (butterfly-type
spectra for ultracold atoms in optical Kagome lattices, see
\cite{Hou09,KR-L14,HKR-L16} and references therein). We refer the
readers also to \cite{Mek03} for historical information on the name
``Kagome'' and how the scientific community became interested in this
structure.

\begin{figure}[ht]

\begin{subfigure}{.45\textwidth}
\begin{tikzpicture}[scale = 1.5]


   \draw[fill = black] (60:1) circle (2pt);
   \draw[fill = black] (120:1) circle (2pt);
   \draw[fill = black] (180:1) circle (2pt);
   
   \begin{scope}[xshift = 2cm]
    \draw[fill = white] (120:1) circle (2pt);
    \draw[fill = white] (180:1) circle (2pt);
   \end{scope}
   \begin{scope}[xshift = -1cm , yshift = 2 * 0.86602540378 * 1 cm]
    \draw[fill = white] (0:1) circle (2pt);
    \draw[fill = white] (240:1) circle (2pt);
   \end{scope}
   \draw[fill = white] (240:1) circle (2pt);
   \begin{scope}[xshift = -2cm]
    \draw[fill = white] (300:1) circle (2pt);
   \end{scope}

    \draw[thick] (7:.94) -- (53:.94);
    \draw[thick] (67:.94) -- (113:.94);
    \draw[thick] (127:.94) -- (173:.94);
    \draw[thick] (187:.94) -- (233:.94);

\begin{scope}[xshift = 1cm, yshift = 2 * 0.86602540378 * 1 cm]
    \draw[thick] (187:.94) -- (233:.94);
    \draw[thick] (247:.94) -- (293:.94);
\end{scope}
\begin{scope}[xshift = -1cm, yshift = 2 * 0.86602540378 * 1 cm]
    \draw[thick] (247:.94) -- (293:.94);
    \draw[thick] (307:.94) -- (353:.94);
\end{scope}
\begin{scope}[xshift = -2cm]
    \draw[thick] (7:.94) -- (53:.94);
    \draw[thick] (307:.94) -- (353:.94);
\end{scope}

   \draw (60:.7) node {$a$};
   \draw (120:.7) node {$b$};
   \draw (180:.85) node {$c$};

   \begin{scope}[xshift = 2cm]
    \draw (110:1.25) node {$(b + \omega_1)$};
    \draw (195:1) node {$(c + \omega_1)$};
   \end{scope}
   \begin{scope}[xshift = -1cm , yshift = 2 * 0.86602540378 * 1 cm]
    \draw (10:1) node {$(c + \omega_2)$};
   \end{scope}
   \draw (250:1.25) node {$(a - \omega_2)$};
   \begin{scope}[xshift = -2cm]
    \draw (70:1.25) node {$(a - \omega_1)$};
    \draw (290:1.25) node {$( b - \omega_2)$};
   \end{scope}
\end{tikzpicture}
\end{subfigure}
\begin{subfigure}{.5\textwidth}
 \begin{tikzpicture}
  \clip ({-.5*\linewidth},-2.5) rectangle ({.5*\linewidth},2.5);
  \draw (0,0) node {\includegraphics[width=\linewidth]{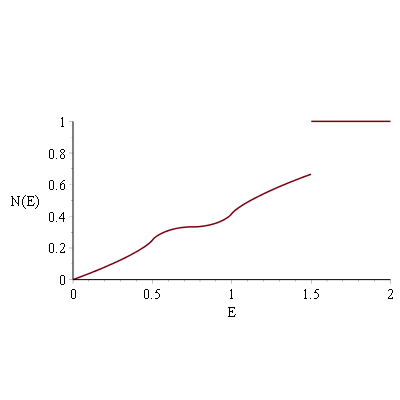}};
 \end{tikzpicture}
\end{subfigure}

 \caption{Fundamental domain of the $(3.6)^2$ tiling (left) and its IDS (right)}
 \label{fig:kagome}
\end{figure}
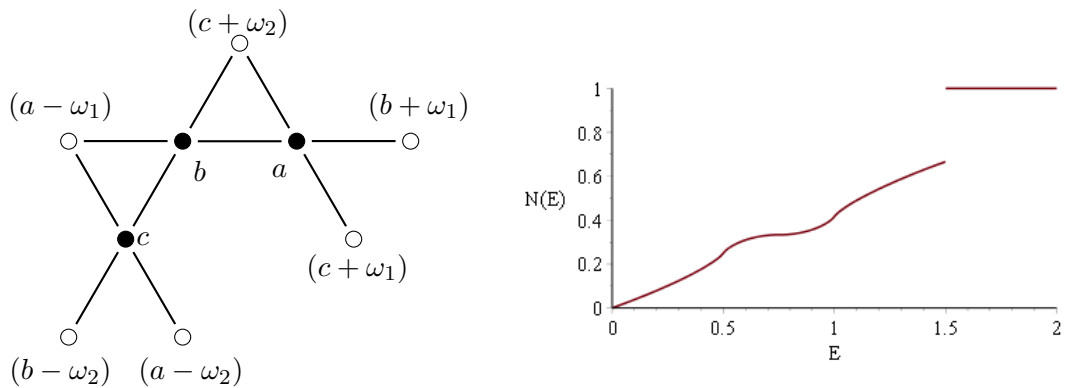

We would like to point out that jumps and strict monotonicity
properties of the IDS of the combinatorial Laplacian on the Kagome
lattice were already determined in \cite[Prop. 3.3]{LPPV09}. We now
derive an explicit formula for the IDS.

A fundamental domain of the Kagome lattice consists of three points,
cf. Figure~\ref{fig:kagome}.  This leads to the matrix
\[
 \Delta^\theta
 =
 \operatorname{Id}
 -
 \frac{1}{4}
 \begin{pmatrix}
 0 & (1 + \euler^{i \theta_1}) & (\euler^{i \theta_1} + \euler^{i \theta_2})\\
 (1 + \euler^{- i \theta_1}) & 0 & (1 + \euler^{i\theta_2})\\ 
 (\euler^{- i \theta_1} + \euler^{- i \theta_2}) &  (1 + \euler^{- i\theta_2}) & 0\\ 
 \end{pmatrix}
\]
with eigenvalues 
\[
 \begin{cases}
  \lambda_{(3.6)^2,1}^\theta
  &=
  \frac{3 - \sqrt{3 + 2 F(u,v)}}{4},\\
  \lambda_{(3.6)^2,2}^\theta
  &=
  \frac{3 + \sqrt{3 + 2 F(u,v)}}{4},\\
  \lambda_{(3.6)^2,3}^\theta
  &=
  \frac{3}{2},\\
 \end{cases}
\]
Furthermore the eigenvalue $3/2$ of $\Delta^\theta$ is $\theta$-independent whence by Theorem~\ref{thm:equivalence}, it corresponds to an infinitely degenerate eigenvalue of $\Delta$.
It can be seen that this eigenvalue is a linear combination of finitely supported eigenvalues on each hexagon where at the vertices of the hexagon, the eigenfunction takes the values $\pm 1$ in alternating order, see also Figure~\ref{fig:eigenfunctions}.
From Lemma~\ref{lem:F}, we deduce:

\begin{proposition}
 
\begin{align*}
  N_{(3.6)^2}(E)
 &=
  \frac{1}{3 (2 \pi)^2} 
  \sum_{k=1}^3
 \operatorname{Vol}
  \left\{
   \theta \in \TT^2 
   \colon
   \lambda_{(3.6)^2,k}^\theta \leq E
  \right\}
  \\
 &=
 \begin{cases}
  0
   &\quad
   \text{if}\ E < 0\\
  \frac{2}{3 \pi^2} 
   \int_{1 - 2 E}^1 
   \frac{ \arccos \left( \frac{4 E^2 - 6 E + 2}{t} - t \right)}{\sqrt{1 - t^2}}
   \mathrm{d}t
   &\quad
   \text{if}\ 0 \leq E < \frac{1}{2}\\
  \frac{1}{3}
   -
   \frac{2}{3 \pi^2} 
   \int_{2 E - 2}^{1 - 2 E} 
   \frac{ \arccos \left( \frac{4 E^2 - 6 E + 2}{t} - t \right)}{\sqrt{1 - t^2}}
   \mathrm{d}t
   &\quad
   \text{if}\ \frac{1}{2} \leq E < \frac{3}{4}\\
  \frac{1}{3}
   +
   \frac{2}{3 \pi^2} 
   \int_{1 - 2 E}^{2 E - 2} 
   \frac{ \arccos \left( \frac{4 E^2 - 6 E + 2}{t} - t \right)}{\sqrt{1 - t^2}}
   \mathrm{d}t
   &\quad
   \text{if}\ \frac{3}{4} \leq E < 1\\
  \frac{2}{3}
   -
   \frac{2}{3 \pi^2} 
   \int_{2 E - 2}^{1} 
   \frac{ \arccos \left( \frac{4 E^2 - 6 E + 2}{t} - t \right)}{\sqrt{1 - t^2}}
   \mathrm{d}t
   &\quad
   \text{if}\ 1 \leq E < \frac{3}{2}\\
   1
   &\quad
   \text{if}\ \frac{3}{2} < E \\
 \end{cases}
\end{align*}
For each hexagon $H$ there exists (up to scalar multiples) exactly one eigenfunction with support on $H$.
Every $\ell^2$-eigenfunction is a linear combination of these special finitely supported eigenfunctions.
\end{proposition}
\subsection{IDS of the $(3.12^2)$ tiling}
The $(3.12^2)$ tiling is the second Archimedean tiling after the $(3.6)^2$ (Kagome) tiling which has compactly supported eigenfunctions.
It also has the interesting feature that the spectrum consists of the two intervals $[0,2/3]$ and $[1, 5/3]$, i.e. it has a proper band structure which might make nanomaterials based on this tiling an interesting candidate for applications.

A fundamental domain consists of six points, $Q = \{ a,b,c,d,e,f \}$, cf. Figure~\ref{fig:super_kagome}. 
We have
\[
 \Delta^\theta 
 =
 \operatorname{Id} 
 -
 \frac{1}{3}
 \begin{pmatrix}
 0 & 1 & 1 & 0 & e^{-i \theta_2} & 0\\
 1 & 0 & 1 & 0 & 0 & e^{i \theta_1}\\
 1 & 1 & 0 & 1 & 0 & 0\\
 0 & 0 & 1 & 0 & 1 & 1\\
 e^{i \theta_2} & 0 & 0 & 1 & 0 & 1\\
 0 & e^{-i \theta_1} & 0 & 1 & 1 & 0\\
 \end{pmatrix}
\]
Its characteristic polynomial is
\begin{align*}
 P_{\Delta^\theta}(E)
 &=
 \frac{(E - 1)(3 E - 5)}{243}
 \big( 81 E^4 - 270 E^3 + 279 E^2 - 90 E -\\
 &\quad - 2 \left( \cos(\theta_1) + \cos(\theta_2) + \cos(\theta_1 - \theta_2) - 3 \right) \big)\\
 &=
 \frac{(E - 1)(3 E - 5)}{243}
 \left( 81 E^4 - 270 E^3 + 279 E^2 - 90 E - 2 F(u,v) + 6 \right),
\end{align*}
where we used again the change of variables $u := (\theta_1 + \theta_2)/\sqrt{2}$, $v = (\theta_1 - \theta_2)/\sqrt{2}$ and the function $F$ from Lemma~\ref{lem:F}.
This is a polynomial of degree $6$ and its roots are
\begin{align*}
 \lambda_{(3.12^2,1)}^\theta
 &=
 \frac{1}{6} \left( 5 - \sqrt{13 + 4 \sqrt{2 F(u,v) + 3}}  \right)
 \in \left[ 0, \frac{5 - \sqrt{13}}{6} \right]
 ,
 \\
 \lambda_{(3.12^2,2)}^\theta
 &=
 \frac{1}{6} \left( 5 - \sqrt{13 - 4 \sqrt{F(u,v) + 3}}  \right)
 \in \left[ \frac{5 - \sqrt{13}}{6} , \frac{2}{3} \right]
 ,
 \\
 \lambda_{(3.12^2,3)}^\theta
 &= 
 1,
 \\ 
 \lambda_{(3.12^2,4)}^\theta
 &=
 \frac{1}{6} \left( 5 + \sqrt{13 - 4 \sqrt{F(u,v) + 3}}  \right)
 \in \left[ 1, \frac{5 + \sqrt{13}}{6} \right]
 ,
 \\
 \lambda_{(3.12^2,5)}^\theta
 &=
 \frac{1}{6} \left( 5 + \sqrt{13 + 4 \sqrt{2 F(u,v) + 3}}  \right)
 \in \left[ \frac{5 + \sqrt{13}}{6}, \frac{5}{3} \right]
 ,
 \\
 \lambda_{(3.12^2,6)}^\theta
 &= 
 \frac{5}{3}.
\end{align*}
\begin{figure}[ht]
\begin{subfigure}{.45\textwidth}
 \begin{tikzpicture}
  \draw[fill = black] (0,0) circle (2pt);
  \draw (.2,0) node {$d$};
  \draw[fill = black] (0,1) circle (2pt);
  \draw (.2,1) node {$c$};
  \draw[fill = black] (-60:1) circle (2pt);
  \draw (-75:1.1) node {$e$};
  \draw[fill = black] (-120:1) circle (2pt);
  \draw (-105:1.15) node {$f$};
  \begin{scope}[yshift = 1cm]
  \draw[fill = black] (60:1) circle (2pt);  
  \draw (75:1.1) node {$b$};
  \draw[fill = black] (120:1) circle (2pt);   
  \draw (105:1.1) node {$a$};
  \end{scope}
  \draw[thick] (0,.2) -- (0,.8);
  \draw[thick] (-60:.8)   -- (-60:.2);
  \draw[thick] (-120:.8)   -- (-120:.2);
  \draw[thick] (-70:.91) -- (-110:.91);
  \begin{scope}[yshift = 1cm]
  \draw[thick] (60:.8)   -- (60:.2);
  \draw[thick] (120:.8)   -- (120:.2);
  \draw[thick] (70:.91) -- (110:.91);
  \end{scope}
  
  \begin{scope}[xshift = .5 cm, yshift = -0.86602540378 * 1 cm]
   \draw [fill = white] (-30:1) circle (2pt);
   \draw (-40:1.2) node {$a - \omega_2$};
   \draw [thick] (-30:.2) -- (-30:.8);
  \end{scope}
  \begin{scope}[xshift = -.5 cm, yshift = -0.86602540378 * 1 cm]
   \draw [fill = white] (-150:1) circle (2pt);
   \draw (-140:1.2) node {$b - \omega_1$};
   \draw [thick] (-150:.2) -- (-150:.8);
  \end{scope}
  \begin{scope}[xshift = .5 cm, yshift = 1.86602540378 * 1 cm]
   \draw [fill = white] (30:1) circle (2pt);
   \draw (40:1.2) node {$f + \omega_1$};
   \draw [thick] (30:.2) -- (30:.8);
  \end{scope}
  \begin{scope}[xshift = -.5 cm, yshift = 1.86602540378 * 1 cm]
   \draw [fill = white] (150:1) circle (2pt);
   \draw [thick] (150:.2) -- (150:.8);
   \draw (140:1.2) node {$e + \omega_2$};
  \end{scope}
\end{tikzpicture}
\end{subfigure}
\begin{subfigure}{.5\textwidth}
 \begin{tikzpicture}
  \clip ({-.5*\linewidth},-2.5) rectangle ({.5*\linewidth},2.5);
  \draw (0,0) node {\includegraphics[width=\linewidth]{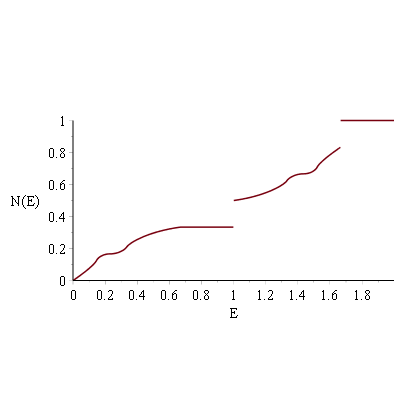}};
 \end{tikzpicture}
\end{subfigure}
\caption{Fundamental domain of the $(3.12^2)$ tiling (left) and its IDS (right)}
\label{fig:super_kagome}
\end{figure}
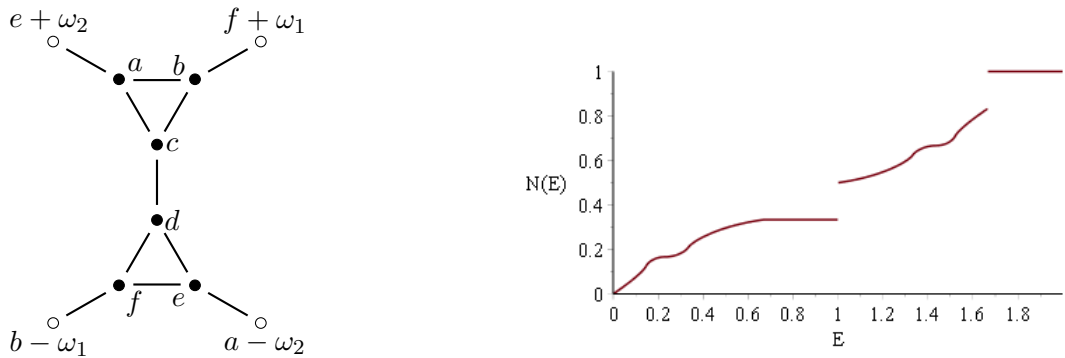
We see that the spectrum of $\Delta$ is supported in the two bands $[0, 2/3]$ and $[1, 5/3]$. 
Furthermore the eigenvalues $1$ and $5/3$ of $\Delta^\theta$ are $\theta$-independent whence by Theorem~\ref{thm:equivalence}, they correspond to two linearly independent, infinitely degenerate supported eigenvalues of $\Delta$.
It can be seen that the corresponding space of eigenfunctions is spanned by functions which are supported on the vertices of a single $12$-gon where cyclically at the vertices of the $12$-gon either the values $1,-1,1,-1,\dots$ (in case $\lambda = 5/3$) or the values $1,1,-1,-1,1,1,\dots$ (in case $\lambda = 1$) appear, see also Figure~\ref{fig:eigenfunctions}.

Using some elementary algebra and Theorem~\ref{thm:IDS_as_integral_expression}, we find
\begin{proposition}
\begin{align*}
 N_{(3.12^2)}(E)
 &=
 \frac{1}{6 (2 \pi)^2}
 \sum_{k = 1}^6
 \vol \left\{ \theta \in \TT^2 \colon \lambda_{(3.12^2,k)}^\theta \leq E \right\}
 \\
 &
  \! \! \! \! 
  \! \! \! \! 
  \! \! \! \! 
  \! \! \! \! 
  =
  \begin{cases}
  0 
  &
  \text{if $E < 0$},\\
  \frac{1}{6 (2 \pi)^2} \vol \left\{ F \geq - \frac{3}{2} + \frac{9}{2} ( 3 E^2 - 5 E + 1)^2 \right\}
  &
  \text{if $0 \leq E < \frac{5 - \sqrt{13}}{6}$},\\
  \frac{1}{3} - \frac{1}{6 (2 \pi)^2} \vol \left\{ F \geq - \frac{3}{2} + \frac{9}{2} ( 3 E^2 - 5 E + 1)^2 \right\}
  &
  \text{if $\frac{5 - \sqrt{13}}{6} \leq E < \frac{2}{3}$},\\
  \frac{1}{3}
  &
  \text{if $\frac{2}{3} \leq E < 1$},\\
  \frac{1}{2} + \frac{1}{6 (2 \pi)^2} \vol \left\{ F \geq - \frac{3}{2} + \frac{9}{2} ( 3 E^2 - 5 E + 1)^2 \right\}
  &
  \text{if $1 \leq E < \frac{5 + \sqrt{13}}{6}$},\\
  \frac{5}{6} - \vol \left\{ F \geq - \frac{3}{2} + \frac{9}{2} ( 3 E^2 - 5 E + 1)^2 \right\}
  &
  \text{if $\frac{5 + \sqrt{13}}{6} \leq E < \frac{5}{3}$},\\
  1
  &
  \text{if $\frac{5}{3} \leq E$}\\
 \end{cases}
 \\
 &
 \! \! \! \! 
 \! \! \! \! 
 \! \! \! \! 
 \! \! \! \! 
 =
 \begin{cases}
  0 
  &
  \text{if $E < 0$},\\[1em]
  \frac{1}{3 \pi^2}
      \int \limits_{- \frac{1}{2} + \frac{3}{2} (3 E^2 - 5 E + 1)}^1
      \frac{\arccos \left( \frac{- \frac{1}{2} + \frac{9}{2} (3 E^2 - 5 E + 1)^2}{2 t} - t \right)}{\sqrt{1 - t^2}}
      \drm t
  &\text{if $0 \leq E < \frac{5 - \sqrt{17}}{6}$},\\[1em]
      \frac{1}{6}
      -
      \frac{1}{3 \pi^2}
      \int \limits_{- \frac{1}{2} - \frac{3}{2} (3 E^2 - 5 E + 1)}^{- \frac{1}{2} + \frac{3}{2} (3 E^2 - 5 E + 1)}
      \frac{\arccos \left( \frac{ - \frac{1}{2} + \frac{9}{2} (3 E^2 - 5 E + 1)^2}{2 t} - t \right)}{\sqrt{1 - t^2}}
      \drm t
  &\text{if $\frac{5 - \sqrt{17}}{6} \leq E < \frac{5 - \sqrt{13}}{6}$},\\[1em]
    \frac{1}{6}
    +
    \frac{1}{3 \pi^2}
    \int \limits_{- \frac{1}{2} - \frac{3}{2} (3 E^2 - 5 E + 1)}^{- \frac{1}{2} + \frac{3}{2} (3 E^2 - 5 E + 1)}
    \frac{\arccos \left( \frac{ - \frac{1}{2} + \frac{9}{2} (3 E^2 - 5 E + 1)^2}{2 t} - t \right)}{\sqrt{1 - t^2}}
    \drm t
   &\text{if $\frac{5 - \sqrt{13}}{6} \leq E < \frac{1}{3}$},\\[1em]
   \frac{1}{3}
   -
   \frac{1}{3 \pi^2}
      \int \limits_{- \frac{1}{2} + \frac{3}{2} (3 E^2 - 5 E + 1)}^1
      \frac{\arccos \left( \frac{- \frac{1}{2} + \frac{9}{2} (3 E^2 - 5 E + 1)^2}{2 t} - t \right)}{\sqrt{1 - t^2}}
      \drm t
   &\text{if $\frac{1}{3} \leq E < \frac{2}{3}$},\\[1em]
   \frac{1}{3}
  &
  \text{if $\frac{2}{3} \leq E < 1$},\\[1em]
  \frac{5}{6} - N_{(3.12^2)} \left( \frac{5}{3} - E \right)
  &
  \text{if $1 \leq E < \frac{5}{3}$},\\[1em]
  1
  &
  \text{if $\frac{5}{3} \leq E$}.
 \end{cases}
\end{align*}
where $F$ is the function explicitly given in Lemma~\ref{lem:F}.
For each $12$-gon $D$, there exist (up to scalar multiples) exactly two linear independent eigenfunctions with support on $D$. 
Every $\ell^2$-eigenfunction is a linear combination of one type of these special finitely supported eigenfunctions.
\end{proposition}

\begin{remark}
 The eigenfunctions on the $(6.3)^2$ and the $(3.12^2)$ are (finite or infinite) linear combinations of eigenfunctions supported on a single hexagon or $12$-gon, respectively, see Figure~\ref{fig:eigenfunctions} for an illustration.
 One observes that both these tesselations share the feature that they contain an $2n$-gon which is either completely surrounded by triangles or where triangles are adjacent to every second edge.
 Since the $(3.6)^2$ tiling and the $(3.12^2)$ tiling are the only ones with this property, this might give an intuitive explaination why exactly these two tilings have finitely supported eigenfunctions.
 However, if one considers periodic graphs which are not based on a tesselation by regular polygons the situation might be different.
 Figure~\ref{fig:non-archimedean_eigenfunction} gives an example of a (non-archimedean) tesselation with finitely supported eigenfunctions.
\end{remark}

\begin{figure}[ht]
 \begin{tikzpicture}
 

\begin{scope}[scale =.6, xshift = -6.5cm]
 
 
 \draw[thick] (30:2) -- (150:2) -- (270:2) -- cycle;
 \draw[thick] (90:2) -- (210:2) -- (330:2) -- cycle;
 
 
 \draw[white, fill=white] (0:1.15470053838) circle (.15cm);
 \draw[white, fill=white] (60:1.15470053838) circle (.15cm);
 \draw[white, fill=white] (120:1.15470053838) circle (.15cm);
 \draw[white, fill=white] (180:1.15470053838) circle (.15cm);
 \draw[white, fill=white] (240:1.15470053838) circle (.15cm);
 \draw[white, fill=white] (300:1.15470053838) circle (.15cm);
 
 \fill (0:1.15470053838) circle (.1cm);
 \fill (60:1.15470053838) circle (.1cm);
 \fill (120:1.15470053838) circle (.1cm);
 \fill (180:1.15470053838) circle (.1cm);
 \fill (240:1.15470053838) circle (.1cm);
 \fill (300:1.15470053838) circle (.1cm);
 
 
 \draw[white, fill=white] (30:2) circle (.15cm);
 \draw[white, fill=white] (90:2) circle (.15cm);
 \draw[white, fill=white] (150:2) circle (.15cm);
 \draw[white, fill=white] (210:2) circle (.15cm);
 \draw[white, fill=white] (270:2) circle (.15cm);
 \draw[white, fill=white] (330:2) circle (.15cm);
 
 \fill (30:2) circle (.1cm);
 \fill (90:2) circle (.1cm);
 \fill (150:2) circle (.1cm);
 \fill (210:2) circle (.1cm);
 \fill (270:2) circle (.1cm);
 \fill (330:2) circle (.1cm); 
 
 
 \draw (0:.7) node {$1$};
 \draw (60:.7) node {$-1$};
 \draw (120:.7) node {$1$};
 \draw (180:.65) node {$-1$};
 \draw (240:.7) node {$1$};
 \draw (300:.7) node {$-1$};
 
 \draw (30:2.4) node {$0$};
 \draw (90:2.4) node {$0$};
 \draw (150:2.4) node {$0$};
 \draw (210:2.4) node {$0$};
 \draw (270:2.4) node {$0$};
 \draw (330:2.4) node {$0$};
  
 \end{scope}


 \begin{scope}[scale =.6]

 
  \draw[thick] (45:2) -- (75:2);
  \draw[thick] (165:2) -- (195:2);
  \draw[thick] (285:2) -- (315:2);
 
 \begin{scope}[rotate=120]
 \draw[thick] (-45:2) -- (-15:2) -- (15:2) -- (45:2);
 \draw[thick] (-45:2) -- (-30:2.82842712475) -- (-15:2);
 \draw[thick] (45:2) -- (30:2.82842712475) -- (15:2);

 \draw[white, fill=white] (-45:2) circle (.15cm);
 \draw[white, fill=white] (-15:2) circle (.15cm);
 \draw[white, fill=white] (15:2) circle (.15cm);
 \draw[white, fill=white] (45:2) circle (.15cm);
 \draw[white, fill=white] (-30:2.82842712475) circle (.15cm);
 \draw[white, fill=white] (30:2.82842712475) circle (.15cm);
 
 \fill (-45:2) circle (.1cm);
 \fill (-15:2) circle (.1cm);
 \fill (15:2) circle (.1cm);
 \fill (45:2) circle (.1cm);
 \fill (-30:2.82842712475) circle (.1cm);
 \fill (30:2.82842712475) circle (.1cm);
 
 \draw (-45:1.5) node {$1$};
 \draw (-15:1.5) node {$-1$};
 \draw (15:1.5) node {$1$};
 \draw (45:1.5) node {$-1$};
 \draw (-30:3.3) node {$0$};
 \draw (30:3.3) node {$0$};
 \end{scope}
 
 \begin{scope}[rotate=0]
 \draw[thick] (-45:2) -- (-15:2) -- (15:2) -- (45:2);
 \draw[thick] (-45:2) -- (-30:2.82842712475) -- (-15:2);
 \draw[thick] (45:2) -- (30:2.82842712475) -- (15:2);

 \draw[white, fill=white] (-45:2) circle (.15cm);
 \draw[white, fill=white] (-15:2) circle (.15cm);
 \draw[white, fill=white] (15:2) circle (.15cm);
 \draw[white, fill=white] (45:2) circle (.15cm);
 \draw[white, fill=white] (-30:2.82842712475) circle (.15cm);
 \draw[white, fill=white] (30:2.82842712475) circle (.15cm);
 
 \fill (-45:2) circle (.1cm);
 \fill (-15:2) circle (.1cm);
 \fill (15:2) circle (.1cm);
 \fill (45:2) circle (.1cm);
 \fill (-30:2.82842712475) circle (.1cm);
 \fill (30:2.82842712475) circle (.1cm);
 
 \draw (-45:1.5) node {$1$};
 \draw (-15:1.5) node {$-1$};
 \draw (15:1.5) node {$1$};
 \draw (45:1.5) node {$-1$};
 \draw (-30:3.3) node {$0$};
 \draw (30:3.3) node {$0$};
 \end{scope}

 \begin{scope}[rotate=240]
 \draw[thick] (-45:2) -- (-15:2) -- (15:2) -- (45:2);
 \draw[thick] (-45:2) -- (-30:2.82842712475) -- (-15:2);
 \draw[thick] (45:2) -- (30:2.82842712475) -- (15:2);

 \draw[white, fill=white] (-45:2) circle (.15cm);
 \draw[white, fill=white] (-15:2) circle (.15cm);
 \draw[white, fill=white] (15:2) circle (.15cm);
 \draw[white, fill=white] (45:2) circle (.15cm);
 \draw[white, fill=white] (-30:2.82842712475) circle (.15cm);
 \draw[white, fill=white] (30:2.82842712475) circle (.15cm);
 
 \fill (-45:2) circle (.1cm);
 \fill (-15:2) circle (.1cm);
 \fill (15:2) circle (.1cm);
 \fill (45:2) circle (.1cm);
 \fill (-30:2.82842712475) circle (.1cm);
 \fill (30:2.82842712475) circle (.1cm);
 
 \draw (-45:1.5) node {$1$};
 \draw (-15:1.5) node {$-1$};
 \draw (15:1.5) node {$1$};
 \draw (45:1.5) node {$-1$};
 \draw (-30:3.3) node {$0$};
 \draw (30:3.3) node {$0$};
 \end{scope}  
 \end{scope}


 \begin{scope}[scale =.6, xshift = 7cm]

 
  \draw[thick] (45:2) -- (75:2);
  \draw[thick] (165:2) -- (195:2);
  \draw[thick] (285:2) -- (315:2);
 
 \begin{scope}[rotate=120]
 \draw[thick] (-45:2) -- (-15:2) -- (15:2) -- (45:2);
 \draw[thick] (-45:2) -- (-30:2.82842712475) -- (-15:2);
 \draw[thick] (45:2) -- (30:2.82842712475) -- (15:2);

 \draw[white, fill=white] (-45:2) circle (.15cm);
 \draw[white, fill=white] (-15:2) circle (.15cm);
 \draw[white, fill=white] (15:2) circle (.15cm);
 \draw[white, fill=white] (45:2) circle (.15cm);
 \draw[white, fill=white] (-30:2.82842712475) circle (.15cm);
 \draw[white, fill=white] (30:2.82842712475) circle (.15cm);
 
 \fill (-45:2) circle (.1cm);
 \fill (-15:2) circle (.1cm);
 \fill (15:2) circle (.1cm);
 \fill (45:2) circle (.1cm);
 \fill (-30:2.82842712475) circle (.1cm);
 \fill (30:2.82842712475) circle (.1cm);
 
 \draw (-45:1.5) node {$1$};
 \draw (-15:1.5) node {$-1$};
 \draw (15:1.5) node {$-1$};
 \draw (45:1.5) node {$1$};
 \draw (-30:3.3) node {$0$};
 \draw (30:3.3) node {$0$};
 \end{scope}
 
 \begin{scope}[rotate=0]
 \draw[thick] (-45:2) -- (-15:2) -- (15:2) -- (45:2);
 \draw[thick] (-45:2) -- (-30:2.82842712475) -- (-15:2);
 \draw[thick] (45:2) -- (30:2.82842712475) -- (15:2);

 \draw[white, fill=white] (-45:2) circle (.15cm);
 \draw[white, fill=white] (-15:2) circle (.15cm);
 \draw[white, fill=white] (15:2) circle (.15cm);
 \draw[white, fill=white] (45:2) circle (.15cm);
 \draw[white, fill=white] (-30:2.82842712475) circle (.15cm);
 \draw[white, fill=white] (30:2.82842712475) circle (.15cm);
 
 \fill (-45:2) circle (.1cm);
 \fill (-15:2) circle (.1cm);
 \fill (15:2) circle (.1cm);
 \fill (45:2) circle (.1cm);
 \fill (-30:2.82842712475) circle (.1cm);
 \fill (30:2.82842712475) circle (.1cm);
 
 \draw (-45:1.5) node {$1$};
 \draw (-15:1.5) node {$-1$};
 \draw (15:1.5) node {$-1$};
 \draw (45:1.5) node {$1$};
 \draw (-30:3.3) node {$0$};
 \draw (30:3.3) node {$0$};
 \end{scope}

 \begin{scope}[rotate=240]
 \draw[thick] (-45:2) -- (-15:2) -- (15:2) -- (45:2);
 \draw[thick] (-45:2) -- (-30:2.82842712475) -- (-15:2);
 \draw[thick] (45:2) -- (30:2.82842712475) -- (15:2);

 \draw[white, fill=white] (-45:2) circle (.15cm);
 \draw[white, fill=white] (-15:2) circle (.15cm);
 \draw[white, fill=white] (15:2) circle (.15cm);
 \draw[white, fill=white] (45:2) circle (.15cm);
 \draw[white, fill=white] (-30:2.82842712475) circle (.15cm);
 \draw[white, fill=white] (30:2.82842712475) circle (.15cm);
 
 \fill (-45:2) circle (.1cm);
 \fill (-15:2) circle (.1cm);
 \fill (15:2) circle (.1cm);
 \fill (45:2) circle (.1cm);
 \fill (-30:2.82842712475) circle (.1cm);
 \fill (30:2.82842712475) circle (.1cm);
 
 \draw (-45:1.5) node {$1$};
 \draw (-15:1.5) node {$-1$};
 \draw (15:1.5) node {$-1$};
 \draw (45:1.5) node {$1$};
 \draw (-30:3.3) node {$0$};
 \draw (30:3.3) node {$0$};
 \end{scope}  
 \end{scope} 
 
 \end{tikzpicture}
 \caption{Eigenfunction in the $(3.6)^2$ tesselation with support on a single hexagon (left) and the two types of eigenfunctions in the $(3.12^2)$ tesselation with support on a single $12$-gon (center and right).}
 \label{fig:eigenfunctions}
 \end{figure}
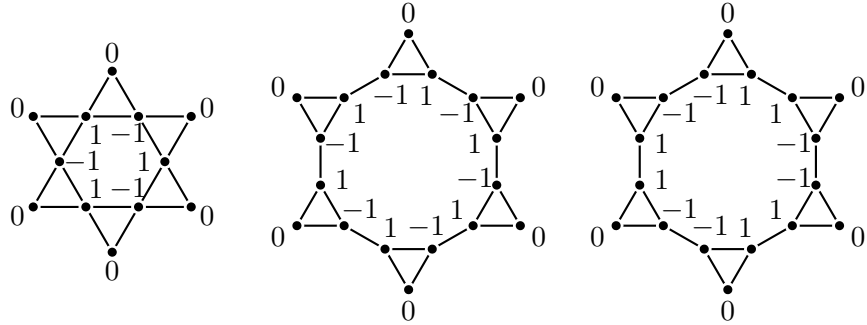

  \begin{figure}
 \begin{tikzpicture}

\begin{scope}
 \draw[thick] (22.5:1) -- 
	(67.5:1) --
	(112.5:1) --
	(157.5:1) --
	(202.5:1) --
	(247.5:1) --
	(292.5:1) --
	(-22.5:1) -- cycle;

\draw[thick] (22.5:1) -- (22.5:2) -- (0:2.2) -- (-22.5:2) -- (-22.5:1);
\draw[thick] (22.5:2) -- (2.2,2.2) -- (67.5:2) -- (67.5:1);

\begin{scope}[rotate=0]
\draw[thick] (22.5:1) -- (22.5:2) -- (0:2.2) -- (-22.5:2) -- (-22.5:1);
\draw[thick] (22.5:2) -- (2.2,2.2) -- (67.5:2) -- (67.5:1);

\draw[white, fill=white] (-22.5:1) circle (.15cm);
\draw[white, fill=white] (-22.5:2) circle (.15cm);
\draw[white, fill=white] (0:2.2) circle (.15cm);
\draw[white, fill=white] (22.5:1) circle (.15cm);
\draw[white, fill=white] (22.5:2) circle (.15cm);
\draw[white, fill=white] (2.2,2.2) circle (.15cm);

\fill (-22.5:1) circle (.075cm);
\fill (-22.5:2) circle (.075cm);
\fill (0:2.2) circle (.075cm);
\fill (22.5:1) circle (.075cm);
\fill (22.5:2) circle (.075cm);
\fill (2.2,2.2) circle (.075cm);

\draw (-22.5:.7) node {$a$};
\draw (22.5:.65) node {$-a$};
\draw (-21:2.4) node {$-1$};
\draw (22.5:2.3) node {$1$};
\draw (0:2.55) node {$0$};
\end{scope}	

\begin{scope}[rotate=90]
\draw[thick] (22.5:1) -- (22.5:2) -- (0:2.2) -- (-22.5:2) -- (-22.5:1);
\draw[thick] (22.5:2) -- (2.2,2.2) -- (67.5:2) -- (67.5:1);

\draw[white, fill=white] (-22.5:1) circle (.15cm);
\draw[white, fill=white] (-22.5:2) circle (.15cm);
\draw[white, fill=white] (0:2.2) circle (.15cm);
\draw[white, fill=white] (22.5:1) circle (.15cm);
\draw[white, fill=white] (22.5:2) circle (.15cm);
\draw[white, fill=white] (2.2,2.2) circle (.15cm);

\fill (-22.5:1) circle (.075cm);
\fill (-22.5:2) circle (.075cm);
\fill (0:2.2) circle (.075cm);
\fill (22.5:1) circle (.075cm);
\fill (22.5:2) circle (.075cm);
\fill (2.2,2.2) circle (.075cm);

\draw (-22.5:.7) node {$a$};
\draw (22.5:.65) node {$-a$};
\draw (-22.5:2.35) node {$-1$};
\draw (22.5:2.3) node {$1$};
\draw (0:2.55) node {$0$};
\end{scope}	

\begin{scope}[rotate=180]
\draw[thick] (22.5:1) -- (22.5:2) -- (0:2.2) -- (-22.5:2) -- (-22.5:1);
\draw[thick] (22.5:2) -- (2.2,2.2) -- (67.5:2) -- (67.5:1);

\draw[white, fill=white] (-22.5:1) circle (.15cm);
\draw[white, fill=white] (-22.5:2) circle (.15cm);
\draw[white, fill=white] (0:2.2) circle (.15cm);
\draw[white, fill=white] (22.5:1) circle (.15cm);
\draw[white, fill=white] (22.5:2) circle (.15cm);
\draw[white, fill=white] (2.2,2.2) circle (.15cm);

\fill (-22.5:1) circle (.075cm);
\fill (-22.5:2) circle (.075cm);
\fill (0:2.2) circle (.075cm);
\fill (22.5:1) circle (.075cm);
\fill (22.5:2) circle (.075cm);
\fill (2.2,2.2) circle (.075cm);

\draw (-22.5:.7) node {$a$};
\draw (22.5:.65) node {$-a$};
\draw (-21:2.4) node {$-1$};
\draw (22.5:2.3) node {$1$};
\draw (0:2.55) node {$0$};
\end{scope}	

\begin{scope}[rotate=270]
\draw[thick] (22.5:1) -- (22.5:2) -- (0:2.2) -- (-22.5:2) -- (-22.5:1);
\draw[thick] (22.5:2) -- (2.2,2.2) -- (67.5:2) -- (67.5:1);

\draw[white, fill=white] (-22.5:1) circle (.15cm);
\draw[white, fill=white] (-22.5:2) circle (.15cm);
\draw[white, fill=white] (0:2.2) circle (.15cm);
\draw[white, fill=white] (22.5:1) circle (.15cm);
\draw[white, fill=white] (22.5:2) circle (.15cm);
\draw[white, fill=white] (2.2,2.2) circle (.15cm);

\fill (-22.5:1) circle (.075cm);
\fill (-22.5:2) circle (.075cm);
\fill (0:2.2) circle (.075cm);
\fill (22.5:1) circle (.075cm);
\fill (22.5:2) circle (.075cm);
\fill (2.2,2.2) circle (.075cm);

\draw (-22.5:.7) node {$a$};
\draw (22.5:.65) node {$-a$};
\draw (-22.5:2.35) node {$-1$};
\draw (22.5:2.3) node {$1$};
\draw (0:2.55) node {$0$};
\end{scope}	


\draw (2.55,2.3) node {$0$};
\draw (-2.55,2.3) node {$0$};
\draw (2.55,-2.3) node {$0$};
\draw (-2.55,-2.3) node {$0$};
	
\draw[white, fill=white] (-22.5:1) circle (.15cm);
\draw[white, fill=white] (-22.5:2) circle (.15cm);
\fill (-22.5:1) circle (.075cm);
\fill (-22.5:2) circle (.075cm);

\end{scope}

 \end{tikzpicture}
 \caption{An example of an elementary cell of a planar, periodic, but non-Archimedean tiling with finitely supported eigenfunctions where the choice $a = 1 \pm \sqrt{2}$ yields an eigenfunctions to the eigenvalue $\lambda = 1 + a/3$.
 }
 
 \label{fig:non-archimedean_eigenfunction}
\end{figure}
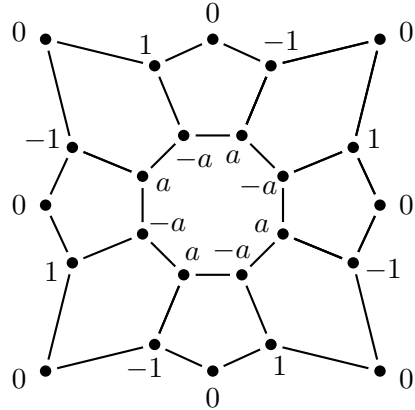

\section{Absence of $\ell^2$-eigenfunctions on the remaining Archimedean tilings}
\label{sec:4}

We show in this section that the remaining Archimedean Tilings, namely
$(3^3.4^2)$, $(4.8^2)$, $(3^2.4.3.4)$, $(3.4.6.4)$, $(4.6.12)$, and
$(3^4.6)$ do not have $\ell^2$-eigenfunctions.  Therefore, their IDS'
are continuous whence -- in the light of the discussion in Remark~\ref{rem:IDS_spectral_measure} -- their have purely absolutely continuous spectrum.
Sufficient geometric conditions for the absence of
finitely supported eigenfunctions in plane tessellations, based on
combinatorial curvature, were given in \cite{KLPS06,Kel11} (see also
\cite{PTV17} about the topic of finitely supported eigenfunctions and
unique continuation). These curvature conditions are not satisfied in the
examples under consideration, so we need to employ Theorem
\ref{thm:equivalence} instead. Since we do not always have explicit
expressions of the eigenvalues of the operators $\Delta^\theta$ or the
volumes of their sublevels sets are too difficult to handle, we will
not provide explicit integral expressions for these IDS', but we are
still able to exclude the existence of $\ell^2$-eigenfunctions.
In fact, for each tiling, we will find the $\theta$-dependent matrix $\Delta^\theta$, make two choices $\theta, \theta' \in \TT^2$, and see that the sets of eigenvalues of $\Delta^\theta$ and $\Delta^{\theta'}$ are disjoint.

\subsection{IDS of the $(3^3.4^2)$ tiling}

A fundamental domain consists of two points $\{a,b\}$ as in Figure~\ref{fig:nikolaushaus}.
This leads to the matrix
\[
 \Delta^\theta 
 =
 \operatorname{Id} 
 -
 \frac{1}{5}
 \begin{pmatrix}
  e^{i \theta_1} + e^{-i \theta_1} & 1 + e^{i \theta_2} + e^{i (\theta_2 - \theta_1)}\\
  1 + e^{- i \theta_2} + e^{- i (\theta_2 - \theta_1)} & e^{i \theta_1} + e^{-i \theta_1}\\
 \end{pmatrix}.
\]
with eigenvalues
\begin{align*}
 \lambda_{\pm} 
 &= 
 1 - \frac{2}{5} \cos(\theta_1) \pm \frac{1}{5} \lvert 1 + e^{i(\theta_1 - \theta_2)} + e^{i \theta_2} \rvert\\
 &=
 1 - \frac{2}{5} \cos(\theta_1) \pm \frac{1}{5} \sqrt{3 + 2 \cos(\theta_1) + 2 \cos(\theta_2) + 2 \cos(\theta_2 - \theta_1)}.
\end{align*}
\begin{figure}[ht]
 \begin{tikzpicture}
  \draw[fill = black] (0,0) circle (2pt);
  \draw[fill = white] (0:1) circle (2pt);
  \draw[fill = white] (60:1) circle (2pt);
  \draw[fill = white] (120:1) circle (2pt);
  \draw[fill = white] (180:1) circle (2pt);
  
  \draw[thick] (0:.2) -- (0:.8);
  \draw[thick] (60:.2) -- (60:.8);
  \draw[thick] (120:.2) -- (120:.8);
  \draw[thick] (180:.2) -- (180:.8);
  
  \draw (.2,-.2) node {$a$};
  \draw (0:1.7) node {$a + \omega_1$};
  \draw (60:1.6) node {$b + \omega_2$};
  \draw (120:1.3) node {$b + \omega_2 - \omega_1$};
  \draw (180:1.6) node {$a - \omega_1$};
%
  
  \begin{scope}[yshift = -1cm]
  \draw[fill = black] (0,0) circle (2pt);
  \draw[fill = white] (180:1) circle (2pt);
  \draw[fill = white] (240:1) circle (2pt);
  \draw[fill = white] (300:1) circle (2pt);
  \draw[fill = white] (360:1) circle (2pt);
  
  \draw[thick] (180:.2) -- (180:.8);
  \draw[thick] (240:.2) -- (240:.8);
  \draw[thick] (300:.2) -- (300:.8);
  \draw[thick] (360:.2) -- (360:.8);
  
  \draw (.2,.2) node {$b$};
  \draw (180:1.6) node {$b - \omega_1$};
  \draw (240:1.6) node {$a - \omega_2$};
  \draw (300:1.3) node {$a - \omega_2 + \omega_1$};
  \draw (360:1.7) node {$b + \omega_1$};
  \end{scope}

  \draw[thick] (0,-.2) -- (0,-.8);
 \end{tikzpicture}
 \caption{Fundamental domain of the $(3^3.4^2)$ tiling}
\label{fig:nikolaushaus}
\end{figure}
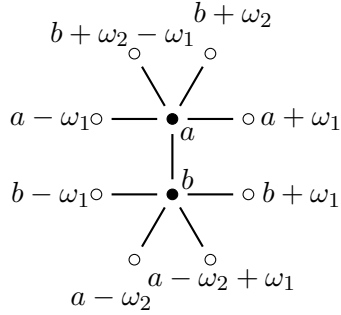
Plugging in $\theta = (0,0)$ and $\theta = (0, \pi)$, we find
\begin{align*}
 \sigma \left( \Delta^{(0,0)} \right)
 =
 \left\{
  0, \frac{6}{5}
 \right\}
 \quad
 \text{and}
 \quad
 \sigma \left( \Delta^{(0,\pi)} \right)
 =
 \left\{
  \frac{2}{5}, \frac{4}{5}
 \right\}.
\end{align*}
Since these sets are disjoint, Theorem~\ref{thm:equivalence} and Corollary~\ref{cor:IDS_continuous} imply
\begin{proposition}
 The $(3^3.4^2)$ tiling has no $\ell^2(\V)$-eigenfunctions.
\end{proposition}

\subsection{IDS of the $(4.8^2)$ tiling}

A fundamental domain consists of the four vertices $\{ a,b,c,d \}$ adjacent to a square, cf. Figure~\ref{fig:4.8^2}.
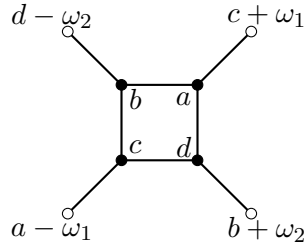
\begin{figure}[ht]
 \begin{tikzpicture}
 \draw[fill = black] (45:0.70710678118) circle (2pt);
 \draw[fill = black] (135:0.70710678118) circle (2pt);
 \draw[fill = black] (225:0.70710678118) circle (2pt);
 \draw[fill = black] (315:0.70710678118) circle (2pt); 
 
 \draw (45:0.45) node {$a$};
 \draw (135:0.45) node {$b$};
 \draw (225:0.45) node {$c$};
 \draw (315:0.45) node {$d$};
 
 \draw[thick] (45:0.70710678118) -- (135:0.70710678118) -- (225:0.70710678118) -- (315:0.70710678118) -- (45:0.70710678118);
  \draw[thick] (45:0.70710678118) -- (45:1.70710678118);
  \draw[thick] (135:0.70710678118) -- (135:1.70710678118);
  \draw[thick] (225:0.70710678118) -- (225:1.70710678118);
  \draw[thick] (315:0.70710678118) -- (315:1.70710678118);
  
 \draw[fill = white] (45:1.70710678118) circle (2pt);
 \draw[fill = white] (135:1.70710678118) circle (2pt);
 \draw[fill = white] (225:1.70710678118) circle (2pt);
 \draw[fill = white] (315:1.70710678118) circle (2pt);
 
 \draw (45:2) node {$c + \omega_1$};
 \draw (135:2) node {$d - \omega_2$};
 \draw (225:2) node {$a - \omega_1$};
 \draw (315:2) node {$b + \omega_2$};
 \end{tikzpicture}
 \caption{Fundamental domain of the $(4.8^2)$ tiling}
\label{fig:4.8^2}
\end{figure}
It leads to the matrix
\[
 \Delta^\theta 
 =
 \operatorname{Id}
 -
 \frac{1}{3}
 \begin{pmatrix}
 0 & 1 & e^{i \theta_1} & 1\\
 1 & 0 & 1 & e^{- i \theta_2}\\
 e^{ - i \theta_1} & 1 & 0 & 1 \\
 1 & e^{ i \theta_2} & 1 & 0 \\ 
 \end{pmatrix}
\] 
Inserting the values $\theta = (0,0)$ and $\theta = (\pi, \pi)$, we find
\begin{align*}
 \sigma \left( \Delta^{(0,0)} \right)
 =
 \left\{ 
 0, \frac{4}{3}
 \right\}
 \quad
 \text{and}
 \quad
 \sigma \left( \Delta^{(\pi, \pi)} \right)
 =
 \left\{
 \frac{2}{3}, 2
 \right\}.
\end{align*}
Since the spectra are disjoint, Theorem~\ref{thm:equivalence} and Corollary~\ref{cor:IDS_continuous}   imply
\begin{proposition}
 The $(4.8^2)$ tiling has no $\ell^2(\V)$-eigenfunctions.
\end{proposition}

\subsection{IDS of the $(3^2.4.3.4)$ tiling}

\begin{figure}[ht]
 \begin{tikzpicture}
 \draw[fill = black] (45:0.70710678118) circle (2pt);
 \draw[fill = black] (135:0.70710678118) circle (2pt);
 \draw[fill = black] (225:0.70710678118) circle (2pt);
 \draw[fill = black] (315:0.70710678118) circle (2pt); 
  \draw[thick] (45:0.70710678118) -- (135:0.70710678118) -- (225:0.70710678118) -- (315:0.70710678118) -- (45:0.70710678118);
  
 \draw (45:0.45) node {$a$};
 \draw (135:0.45) node {$b$};
 \draw (225:0.45) node {$c$};
 \draw (315:0.45) node {$d$};


\begin{scope}[rotate = 0]
 \draw[thick] (135:0.70710678118) -- (0,1.36602540378) -- (45:0.70710678118) -- (1,1.36602540378) -- (0,1.36602540378);
 \draw[fill = white] (0,1.36602540378) circle (2pt);
 \draw[fill = white] (1,1.36602540378) circle (2pt);
 \draw (-.2,1.6) node {$c + \omega_1$};
 \draw (1.2,1.6) node {$d + \omega_1$};
\end{scope}

\begin{scope}[rotate = 270]
 \draw[thick] (135:0.70710678118) -- (0,1.36602540378) -- (45:0.70710678118) -- (1,1.36602540378) -- (0,1.36602540378);
 \draw[fill = white] (0,1.36602540378) circle (2pt);
 \draw[fill = white] (1,1.36602540378) circle (2pt);
 \draw (0,2) node {$b + \omega_2$};
 \draw (1,2) node {$c + \omega_2$};
\end{scope}

\begin{scope}[rotate = 180]
 \draw[thick] (135:0.70710678118) -- (0,1.36602540378) -- (45:0.70710678118) -- (1,1.36602540378) -- (0,1.36602540378);
 \draw[fill = white] (0,1.36602540378) circle (2pt);
 \draw[fill = white] (1,1.36602540378) circle (2pt);
 \draw (-.2,1.6) node {$a - \omega_1$};
 \draw (1.2,1.6) node {$b - \omega_1$};
\end{scope}

\begin{scope}[rotate = 90]
 \draw[thick] (135:0.70710678118) -- (0,1.36602540378) -- (45:0.70710678118) -- (1,1.36602540378) -- (0,1.36602540378);
 \draw[fill = white] (0,1.36602540378) circle (2pt);
 \draw[fill = white] (1,1.36602540378) circle (2pt);
 \draw (0,2) node {$d - \omega_2$};
 \draw (1,2) node {$a - \omega_2$}; 
\end{scope}
 \end{tikzpicture}
 \caption{Fundamental domain of the $(3^2.4.3.4)$ tiling}
\label{fig:3^2.4.3.4}
\end{figure}
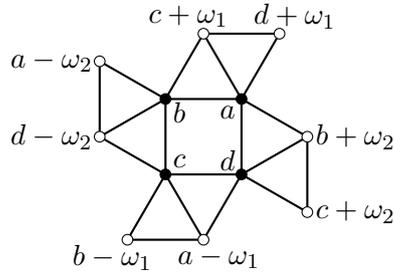

A fundamental domain consists of the four vertices $\{ a,b,c,d \}$ adjacent to a square with edges parallel to the axes, cf. Figure~\ref{fig:3^2.4.3.4}.
It leads to the matrix
\[
 \Delta^\theta 
 =
 \operatorname{Id}
 -
 \frac{1}{5}
 \begin{pmatrix}
 0 & 1 + e^{i \theta_2} & e^{i \theta_1} & 1 + e^{i \theta_1}\\
 1 + e^{- i \theta_2} & 0 & 1 + e^{i \theta_1} & e^{- i \theta_2}  \\
 e^{- i \theta_1} & 1 + e^{- i \theta_1} & 0 & 1 + e^{- i \theta_2}  \\
 1 + e^{- i \theta_1} & e^{i \theta_2} & 1 + e^{i \theta_2} & 0 \\
 \end{pmatrix}.
\] 
Inserting $\theta = (0,0)$ and $\theta = (\pi,0)$, we find
\[
 \sigma \left( \Delta^{(0,0)} \right)
 =
 \left\{ 
  0, \frac{6}{5}, \frac{8}{5}
 \right\}
 \quad
 \text{and}
 \quad
 \sigma \left( \Delta^{(\pi,0)} \right)
 =
 \left\{
 1 - \sqrt{5}^{-1}, 1 + \sqrt{5}^{-1}
 \right\}
 .
\]
Again, these sets are disjoint whence Theorem~\ref{thm:equivalence} and Corollary~\ref{cor:IDS_continuous} imply
\begin{proposition}
 The $(3^2.4.3.4)$-tiling has no $\ell^2(\V)$-eigenfunctions.
\end{proposition}

\subsection{IDS of the $3.4.6.4$ tiling}

\begin{figure}[ht]
 \begin{tikzpicture}
 \draw[fill = black] (0:1) circle (2pt);
 \draw[fill = black] (60:1) circle (2pt);
 \draw[fill = black] (120:1) circle (2pt);
 \draw[fill = black] (180:1) circle (2pt);
 \draw[fill = black] (240:1) circle (2pt);
 \draw[fill = black] (300:1) circle (2pt);
 \draw[thick] (0:1) -- (60:1) -- (120:1) -- (180:1) -- (240:1) -- (300:1) -- (0:1);

 \draw (0:0.7) node {$a$};
 \draw (60:0.7) node {$b$};
 \draw (120:0.7) node {$c$};
 \draw (180:0.7) node {$d$};
 \draw (240:0.7) node {$e$};
 \draw (300:0.7) node {$f$};
%
\begin{scope}[rotate = 0]
 \draw[thick] (120:1) -- (-.5,1.86602540378) -- (.5,1.86602540378)  -- (60:1);
 \draw[fill = white] (-.5,1.86602540378) circle (2pt);
 \draw[fill = white] (.5,1.86602540378) circle (2pt);
 \draw (-.75,2.2) node {$e + \omega_2$};
 \draw (.75,2.2) node {$f + \omega_2$};
\end{scope}

\begin{scope}[rotate = 60]
 \draw[thick] (120:1) -- (-.5,1.86602540378) -- (.5,1.86602540378)  -- (60:1);
 \draw[fill = white] (-.5,1.86602540378) circle (2pt);
 \draw[fill = white] (.5,1.86602540378) circle (2pt);
 \draw (-1.1,2.9) node {$f + \omega_2 - \omega_1$};
 \draw (.6,2.1) node {$a + \omega_2 - \omega_1$};
\end{scope}

\begin{scope}[rotate = 120]
 \draw[thick] (120:1) -- (-.5,1.86602540378) -- (.5,1.86602540378)  -- (60:1);
 \draw[fill = white] (-.5,1.86602540378) circle (2pt);
 \draw[fill = white] (.5,1.86602540378) circle (2pt);
 \draw (-.6,2.2) node {$a - \omega_1$};
 \draw (.6,2.5) node {$b - \omega_1$};
\end{scope}

\begin{scope}[rotate = 180]
 \draw[thick] (120:1) -- (-.5,1.86602540378) -- (.5,1.86602540378)  -- (60:1);
 \draw[fill = white] (-.5,1.86602540378) circle (2pt);
 \draw[fill = white] (.5,1.86602540378) circle (2pt);
 \draw (-.75,2.2) node {$b - \omega_2$};
 \draw (.75,2.2) node {$c - \omega_2$};
\end{scope}

\begin{scope}[rotate = 240]
 \draw[thick] (120:1) -- (-.5,1.86602540378) -- (.5,1.86602540378)  -- (60:1);
 \draw[fill = white] (-.5,1.86602540378) circle (2pt);
 \draw[fill = white] (.5,1.86602540378) circle (2pt);
 \draw (-1,3) node {$c + \omega_1 - \omega_2$};
 \draw (.55,2.4) node {$d  + \omega_1- \omega_2$};
\end{scope}

\begin{scope}[rotate = 300]
 \draw[thick] (120:1) -- (-.5,1.86602540378) -- (.5,1.86602540378)  -- (60:1);
 \draw[fill = white] (-.5,1.86602540378) circle (2pt);
 \draw[fill = white] (.5,1.86602540378) circle (2pt);
 \draw (-.6,2.1) node {$d + \omega_1$};
 \draw (.8,2.5) node {$e + \omega_1$};
\end{scope}
\end{tikzpicture}
 \caption{Fundamental domain of the $(3.4.6.4)$ tiling}
\label{fig:3.4.6.4}
\end{figure}
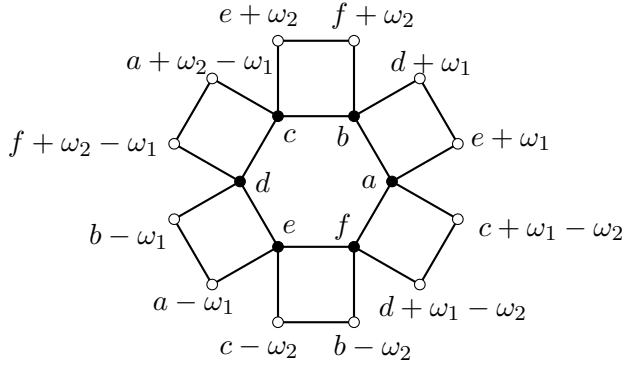

A fundamental domain consists of the six vertices $\{a,b,c,d,e,f \}$ around a hexagon, cf. Figure~\ref{fig:3.4.6.4}.
It leads to the matrix
\[
 \Delta^\theta 
 =
 \operatorname{Id}
 -
 \frac{1}{4}
 \begin{pmatrix}
   0 & 1 & e^{-i(\theta_2 - \theta_1)} & 0 & e^{i \theta_1} & 1 \\
   1 & 0 & 1 & e^{i \theta_1} & 0 & e^{i \theta_2} \\
   e^{i(\theta_2 - \theta_1)} & 1 & 0 & 1 & e^{i \theta_2} & 0 \\
   0 & e^{- i \theta_1} & 1 & 0 & 1 & e^{i(\theta_2 - \theta_1)} \\
   e^{- i \theta_1} & 0 & e^{- i \theta_2} & 1 & 0 & 1 \\
   1 & e^{- i \theta_2} & 0 & e^{-i(\theta_2 - \theta_1)} & 1 & 0\\
 \end{pmatrix}.
\] 
We compare the spectra of $\Delta^\theta$ at $\theta = (0,0)$ and $\theta = (\pi, \pi/2)$: 
\begin{align*}
 \sigma \left( \Delta^{(0,0)} \right)
 &=
 \left\{ 
 0, 1, \frac{3}{2}
 \right\}
 \quad
 \text{and}
 \\
 \sigma \left( \Delta^{(\pi, \pi/2)} \right)
 &=
 \left\{
 \lambda \in \CC 
 \colon
 \lambda^6 - 6 \lambda^5 + \frac{57}{4} \lambda^4 - 17 \lambda^3 + \frac{85}{8} \lambda^2 - \frac{13}{4} \lambda + \frac{95}{256} 
 = 
 0
 \right\}.
\end{align*}
It is straightforward to verify that these sets are disjoint.
By Theorem~\ref{thm:equivalence} and Corollary~\ref{cor:IDS_continuous}, we find
\begin{proposition}
 The $(3.4.6.4)$ tiling has no $\ell^2(\V)$-eigenfunctions.
\end{proposition}

\subsection{IDS of the $(4.6.12)$ tiling}

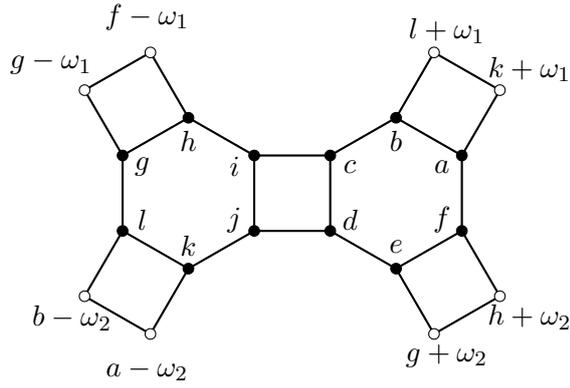
\begin{figure}[ht]
 \begin{tikzpicture}
 \draw[fill = black] (30:1) circle (2pt);
 \draw[fill = black] (90:1) circle (2pt);
 \draw[fill = black] (150:1) circle (2pt);
 \draw[fill = black] (210:1) circle (2pt);
 \draw[fill = black] (270:1) circle (2pt);
 \draw[fill = black] (330:1) circle (2pt);
 \draw[thick] (30:1) -- (90:1) -- (150:1) -- (210:1) -- (270:1) -- (330:1) -- (30:1);

 \draw (30:0.7) node {$a$};
 \draw (90:0.7) node {$b$};
 \draw (150:0.7) node {$c$};
 \draw (210:0.7) node {$d$};
 \draw (270:0.7) node {$e$};
 \draw (330:0.7) node {$f$};
\begin{scope}[rotate = 60]
 \draw[thick] (0.86602540378,-.5) -- (1.86602540378,-.5) -- (1.86602540378,.5) -- (0.86602540378,.5);
 \draw[fill = white] (1.86602540378,-.5) circle (2pt);
 \draw[fill = white] (1.86602540378,.5) circle (2pt);
 \draw (2.3,-.7) node {$k + \omega_1$};
 \draw (2.2,.5) node {$l + \omega_1$};
\end{scope}

\begin{scope}[rotate = -60]
 \draw[thick] (0.86602540378,-.5) -- (1.86602540378,-.5) -- (1.86602540378,.5) -- (0.86602540378,.5);
 \draw[fill = white] (1.86602540378,-.5) circle (2pt);
 \draw[fill = white] (1.86602540378,.5) circle (2pt);
 \draw (2.2,-.5) node {$g + \omega_2$};
 \draw (2.3,.7) node {$h + \omega_2$};
\end{scope}


\begin{scope}[xshift = -2.73205080757cm]

 \draw[thick] (0.86602540378,-.5) -- (1.86602540378,-.5) ;
 \draw[thick] (0.86602540378,.5) -- (1.86602540378,.5);

 \draw[fill = black] (30:1) circle (2pt);
 \draw[fill = black] (90:1) circle (2pt);
 \draw[fill = black] (150:1) circle (2pt);
 \draw[fill = black] (210:1) circle (2pt);
 \draw[fill = black] (270:1) circle (2pt);
 \draw[fill = black] (330:1) circle (2pt);
 \draw[thick] (30:1) -- (90:1) -- (150:1) -- (210:1) -- (270:1) -- (330:1) -- (30:1);

 \draw (30:0.7) node {$i$};
 \draw (90:0.7) node {$h$};
 \draw (150:0.7) node {$g$};
 \draw (210:0.7) node {$l$};
 \draw (270:0.7) node {$k$};
 \draw (330:0.7) node {$j$};
 
\begin{scope}[rotate = 120]
 \draw[thick] (0.86602540378,-.5) -- (1.86602540378,-.5) -- (1.86602540378,.5) -- (0.86602540378,.5);
 \draw[fill = white] (1.86602540378,-.5) circle (2pt);
 \draw[fill = white] (1.86602540378,.5) circle (2pt);
 \draw (2.3,-.7) node {$f - \omega_1$};
 \draw (2.4,.7) node {$g - \omega_1$};
\end{scope}

\begin{scope}[rotate = -120]
 \draw[thick] (0.86602540378,-.5) -- (1.86602540378,-.5) -- (1.86602540378,.5) -- (0.86602540378,.5);
 \draw[fill = white] (1.86602540378,-.5) circle (2pt);
 \draw[fill = white] (1.86602540378,.5) circle (2pt);
 \draw (2.2,-.5) node {$b - \omega_2$};
 \draw (2.3,.7) node {$a - \omega_2$};
\end{scope}
 
\end{scope}

\end{tikzpicture}
 \caption{Fundamental domain of the $(4.6.12)$ tiling}
\label{fig:4.6.12}
\end{figure}

A fundamental domain consists of the $12$ vertices constituting two neighboring hexagons.
This leads to 
\[
 \Delta^\theta
 =
 \operatorname{Id} 
 -
 \frac{1}{3}
 \begin{pmatrix}
  A & B \\
  \overline{B}^T & A \\
 \end{pmatrix},
\]
where
\[
 A 
 =
 \begin{pmatrix}
  0 & 1 & 0 & 0 & 0 & 1 \\
  1 & 0 & 1 & 0 & 0 & 0 \\
  0 & 1 & 0 & 1 & 0 & 0 \\
  0 & 0 & 1 & 0 & 1 & 0 \\
  0 & 0 & 0 & 1 & 0 & 1 \\
  1 & 0 & 0 & 0 & 1 & 0 \\
 \end{pmatrix}
 \quad
 \text{and}
 \quad
 B
 =
 \begin{pmatrix}
  0 & 0 & 0 & 0 & e^{\theta_1} & 0 \\
  0 & 0 & 0 & 0 & 0 & e^{\theta_1} \\
  0 & 0 & 1 & 0 & 0 & 0 \\
  0 & 0 & 0 & 1 & 0 & 0 \\
  e^{\theta_2} & 0 & 0 & 0 & 0 & 0 \\
  0 & e^{\theta_2} & 0 & 0 & 0 & 0 \\
 \end{pmatrix}.
\]
It suffices to study the spectrum of the adjacency matrix
\[
 M^\theta
 :=
 \begin{pmatrix}
  A & B \\
  \overline{B}^T & A \\
 \end{pmatrix}
\]
since the spectrum of $M^\theta$ differs from the spectrum of $\Delta^\theta$ only by an invertible linear affine transformation. 
Thus, we need to check that there are $\theta$, $\theta'$ such that $\sigma(M^\theta) \cap \sigma(M^{\theta'}) = \emptyset$.
Plugging in the values $(0,0)$ and $(\pi, \pi/2)$ for $\theta$, we find
\begin{align*}
 \sigma \left( M^{(0,0)} \right)
 &=
 \left\{ 
  \pm 1, \pm \sqrt{3}, \pm 3
 \right\}
 \quad
 \text{and}\\
 \sigma \left( M^{(\pi,\pi/2)} \right)
 &=
 \left\{
 \lambda \in \CC
 \colon
 \lambda^{12} - 18 \lambda^{10} + 111 \lambda^8 - 268 \lambda^6 + 207 \lambda^4 - 50 \lambda^2 + 1 
 = 
 0
 \right\}
\end{align*}
and again it is straightforward to verify that these sets are disjoint whence also $\sigma(\Delta^{(0,0)}) \cap \sigma(\Delta^{(\pi, \pi/2)}) = \emptyset$.
Theorem~\ref{thm:equivalence} and Corollary~\ref{cor:IDS_continuous} imply
\begin{proposition}
 The $(4.6.12)$-tiling has no $\ell^2(\V)$-eigenfunctions.
\end{proposition}

\subsection{IDS of the $(3^4.6)$ tiling}

\begin{figure}[ht]
 \begin{tikzpicture}
 \draw[fill = black] (0:1) circle (2pt);
 \draw[fill = black] (60:1) circle (2pt);
 \draw[fill = black] (120:1) circle (2pt);
 \draw[fill = black] (180:1) circle (2pt);
 \draw[fill = black] (240:1) circle (2pt);
 \draw[fill = black] (300:1) circle (2pt);
 \draw[thick] (0:1) -- (60:1) -- (120:1) -- (180:1) -- (240:1) -- (300:1) -- (0:1);

 \draw (0:0.7) node {$a$};
 \draw (60:0.7) node {$b$};
 \draw (120:0.7) node {$c$};
 \draw (180:0.7) node {$d$};
 \draw (240:0.7) node {$e$};
 \draw (300:0.7) node {$f$};

\begin{scope}[rotate = 0]
\draw[thick] (60:1) -- (0,1.73205080757) -- (120:1) -- (-1,1.73205080757) -- (0,1.73205080757);
\draw[fill = white] (0,1.73205080757) circle (2pt);
\draw[fill = white] (-1,1.73205080757) circle (2pt);
\draw (0,2.35) node {$f + \omega_1 - \omega_2$};
\draw (-1.5,2) node {$e + \omega_1 - \omega_2$};
\end{scope}

\begin{scope}[rotate = 60]
\draw[thick] (60:1) -- (0,1.73205080757) -- (120:1) -- (-1,1.73205080757) -- (0,1.73205080757);
\draw[fill = white] (0,1.73205080757) circle (2pt);
\draw[fill = white] (-1,1.73205080757) circle (2pt);
\draw (0,2.25) node {$a - \omega_2$};
\draw (-1.25,2.4) node {$f - \omega_2$};
\end{scope}

\begin{scope}[rotate = 120]
\draw[thick] (60:1) -- (0,1.73205080757) -- (120:1) -- (-1,1.73205080757) -- (0,1.73205080757);
\draw[fill = white] (0,1.73205080757) circle (2pt);
\draw[fill = white] (-1,1.73205080757) circle (2pt);
\draw (.25,2.25) node {$b - \omega_1$};
\draw (-.75,2.4) node {$a - \omega_1$};
\end{scope}

\begin{scope}[rotate = 180]
\draw[thick] (60:1) -- (0,1.73205080757) -- (120:1) -- (-1,1.73205080757) -- (0,1.73205080757);
\draw[fill = white] (0,1.73205080757) circle (2pt);
\draw[fill = white] (-1,1.73205080757) circle (2pt);
\draw (0,2.35) node {$c - \omega_1 + \omega_2$};
\draw (-1.5,2) node {$b - \omega_1 + \omega_2$};
\end{scope}

\begin{scope}[rotate = 240]
\draw[thick] (60:1) -- (0,1.73205080757) -- (120:1) -- (-1,1.73205080757) -- (0,1.73205080757);
\draw[fill = white] (0,1.73205080757) circle (2pt);
\draw[fill = white] (-1,1.73205080757) circle (2pt);
\draw (-.2,2.4) node {$d + \omega_2$};
\draw (-1.2,2.4) node {$c + \omega_2$};
\end{scope}

\begin{scope}[rotate = 300]
\draw[thick] (60:1) -- (0,1.73205080757) -- (120:1) -- (-1,1.73205080757) -- (0,1.73205080757);
\draw[fill = white] (0,1.73205080757) circle (2pt);
\draw[fill = white] (-1,1.73205080757) circle (2pt);
\draw (.2,2.4) node {$e + \omega_1$};
\draw (-.8,2.4) node {$d + \omega_1$};
\end{scope}
\end{tikzpicture}
 \caption{Fundamental domain of the $(3^4.6)$ tiling}
\label{fig:3^4.6}
\end{figure}
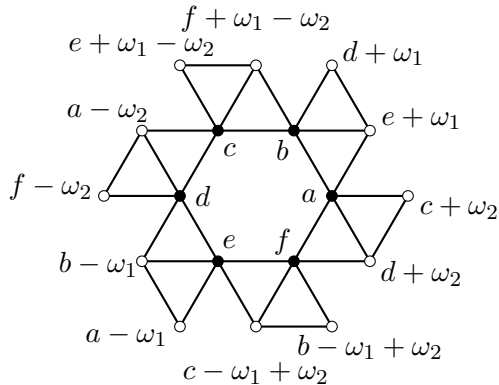

A fundamental domain consists of the six vertices $\{ a,b,c,d,e,f \}$ corresponding to a hexagon, cf. Figure~\ref{fig:3^4.6}.
This leads to the matrix
\[
 \Delta^\theta 
 =
 \operatorname{Id}
 -
 \frac{1}{5}
 \begin{pmatrix}
  0 & 1 & e^{i \theta_2} & e^{i \theta_2} & e^{i \theta_1} & 1 \\
  1 & 0 & 1 & e^{i \theta_1} & e^{i \theta_1} & e^{- i (\theta_2 - \theta_1)} \\
  e^{- i \theta_2} & 1 & 0 & 1 & e^{-i (\theta_2 - \theta_1)} & e^{-i (\theta_2 - \theta_1)} \\
  e^{-i \theta_2} & e^{- i \theta_1} & 1 & 0 & 1 & e^{- i \theta_2} \\
  e^{- i \theta_1} & e^{- i \theta_1} & e^{i (\theta_2 - \theta_1)} & 1 & 0 & 1 \\
  1 & e^{i (\theta_2 - \theta_1)} & e^{i (\theta_2 - \theta_1)} & e^{i \theta_2} & 1 & 0 \\
  \end{pmatrix}.
\] 

We choose the particular values $\theta = (0,0)$ and $\theta = (\pi, \pi/2)$ and find
\begin{align*}
 \sigma \left( \Delta^{(0,0)} \right)
 &=
 \left\{
 0, \frac{6}{5}
 \right\}
 \quad
 \text{and}\\
 \sigma \left( \Delta^{(\pi,\pi/2)} \right)
 &=
 \left\{
 \lambda \in \CC
 \colon
 \lambda^6 
 - 6 \lambda^5 
 + \frac{72}{5} \lambda^4 
 - \frac{2192}{125} \lambda^3 
 + \frac{7056}{625} \lambda^2
 - \frac{11192}{3125} \lambda
 + \frac{6656}{15625}
 = 
 0
 \right\}
 \!
 .
\end{align*}
It is straightforward to verify that these sets are disjoint and by Theorem~\ref{thm:equivalence} and Corollary~\ref{cor:IDS_continuous} we find
\begin{proposition}
 The $(3^4.6)$-tiling has no $\ell^2(\V)$-eigenfunctions.
\end{proposition}

{\bf Acknowledgement:}
The authors would like to thank Ivan Veseli\'c and Christoph Schumacher for helpful discussions.
Parts of this article were written while the authors enjoyed the hospitality of the Isaac Newton Institute during the programme \emph{Non-Positive Curvature Group Actions and Cohomology} and while the first named author enjoyed the hospitality of Technische Universit\"at Dortmund.
The second named author was in part supported by the European Research Council starting grant 639305 (SPECTRUM).

\newcommand{\etalchar}[1]{$^{#1}$}

\end{document}